\documentclass[11pt,a4paper]{article}
\pdfoutput=1

\usepackage[utf8]{inputenc}
\usepackage[T1]{fontenc}
\usepackage{lmodern}

\usepackage[top=1in, bottom=1in, left=1in, right=1in]{geometry}

\usepackage{microtype}

\usepackage{amsmath}
\usepackage{amssymb}
\usepackage{amsfonts,amsthm}
\usepackage{thmtools, thm-restate}

\usepackage{tikz}

\usepackage[lined, algoruled, linesnumbered]{algorithm2e}

\usepackage{hyperref,xcolor}

\definecolor{winered}{rgb}{0.5,0,0}

\hypersetup
{
    pdfborder={0 0 0},
    colorlinks=true,
    linkcolor={winered},
    urlcolor={winered},
    filecolor={winered},
    citecolor={winered},
    linktoc=all,
}

\usetikzlibrary{arrows,decorations.pathmorphing,decorations.shapes,backgrounds,fit}
\pgfkeys{tikz/cc1/.style={dotted} }
\pgfkeys{tikz/cc4/.style={dashed} }
\pgfkeys{tikz/cc3/.style={decorate,decoration=bumps} }
\pgfkeys{tikz/cc2/.style={decorate,decoration=snake} }
\pgfkeys{tikz/solid/.style={line width=2pt} }
\pgfkeys{tikz/solid2/.style={line width=1pt} }
\pgfkeys{tikz/nodostile1/.style={draw,line width=1pt}}
\pgfkeys{tikz/nodostile2/.style={draw,line width=2pt}}
\pgfkeys{tikz/arcostile1/.style={very thick} }
\pgfkeys{tikz/arcostile2/.style={red,very thick} }
\pgfkeys{tikz/arcostile3/.style={thick} }
\pgfkeys{tikz/arcostile4/.style={line width=2pt} }

\newcommand{\ignore}[1]{}

\newtheorem{theorem}{Theorem}[section]
\newtheorem{lemma}[theorem]{Lemma}
\newtheorem{definition}[theorem]{Definition}
\newtheorem{corollary}[theorem]{Corollary}

\newcommand{\condense}{\textsc{Condense}}
\newcommand{\splitop}{\textsc{Split}}

\newcommand{\builddec}{\textsc{Build-SCC-Decomposition}}

\newcommand{\con}{\text{con}}
\newcommand{\R}{\text{R}}
\DeclareMathOperator{\LCA}{LCA}
\DeclareMathOperator{\NCA}{NCA}

\DeclareMathOperator*{\argmax}{arg\,max}

\newcommand{\sidenote}[1]{}

\begin{document}

\title{Decremental Data Structures for \\ Connectivity and Dominators in Directed Graphs\thanks{Accepted to the 44th International Colloquium on Automata, Languages, and Programming (ICALP 2017).}}
\author{Loukas Georgiadis\thanks{Work partially done while visiting University of Rome Tor Vergata.} \\ {\normalsize University of Ioannina} 
  \and Thomas Dueholm Hansen\thanks{Work partially done while visiting University of Rome Tor Vergata. Supported by the Carlsberg Foundation, grant no. CF14-0617.} \\ {\normalsize Aarhus University} 
  \and Giuseppe F. Italiano\thanks{Partially supported by MIUR, the Italian Ministry of Education, University and Research, under Project AMANDA (Algorithmics for MAssive and Networked DAta).} \\ {\normalsize University of Rome Tor Vergata} 
  \and Sebastian Krinninger\thanks{Work partially done while visiting University of Rome Tor Vergata and while at Max Planck Institute for Informatics, Saarland Informatics Campus, Germany.} \\ {\normalsize University of Vienna} 
  \and Nikos Parotsidis \\ {\normalsize University of Rome Tor Vergata}} 
\date{}

%
%
%

\maketitle

\begin{abstract}
We introduce a new dynamic data structure for maintaining the strongly connected components (SCCs) of a directed graph (digraph) under edge deletions, so as to answer a rich repertoire of connectivity queries.
Our main technical contribution is a decremental data structure that supports sensitivity queries of the form ``are $ u $ and $ v $ strongly connected in the graph $ G \setminus w $?'', for any triple of vertices $ u, v, w $, while $ G $ undergoes deletions of edges.
Our data structure processes a sequence of edge deletions in a digraph with $n$ vertices in $O(m n \log{n})$ total time and $O(n^2 \log{n})$ space, where $m$ is the number of edges before any deletion, and answers the above queries in constant time.
We can leverage our data structure to obtain decremental data structures for many more types of queries within the same time and space complexity.
For instance for edge-related queries, such as testing whether two query vertices $u$ and $v$ are strongly connected in $G \setminus e$, for some query edge $e$.

As another important application of our decremental data structure, we provide the first nontrivial algorithm for maintaining the dominator tree of a flow graph under edge deletions.
We present an algorithm that processes a sequence of edge deletions in a flow graph in $O(m n \log{n})$ total time and $O(n^2 \log{n})$ space.
For reducible flow graphs we provide an $O(mn)$-time and $O(m + n)$-space algorithm.
We give a conditional lower bound that provides evidence that these running times may be tight up to subpolynomial factors.

\end{abstract}

\section{Introduction}

Dynamic graph algorithms have been extensively studied for several decades, and many important results have been achieved for dynamic versions of
fundamental problems, including connectivity, 2-edge and 2-vertex connectivity, minimum spanning tree, transitive closure, and shortest paths (see, e.g., the survey in~\cite{EGI09}).
We recall that a dynamic graph problem is said to be \emph{fully dynamic} if it involves both insertions and deletions of edges, \emph{incremental} if it only involves edge insertions, and \emph{decremental} if it only involves edge deletions.

The decremental strongly connected components (SCCs) problem asks us to maintain, under edge deletions in a directed graph $G$, a data structure that given two vertices $u$ and $v$ answers whether $u$ and $v$ are strongly connected in $G$. We extend this problem to sensitvity queries of the form ``are $ u $ and $ v $ strongly connected in the graph $ G \setminus w $?'', for any triple of vertices $ u, v, w $, i.e., we additionally allow the query to temporarily remove a third vertex $w$. We show that this extended decremental SCC problem can be used to answer fast
a rich repertoire of connectivity queries, and we present a new and efficient data structure for the problem. In particular, our data structure for the extended decremental SCC problem can be used to support edge-related queries, such
as maintaining the strong bridges of a digraph, testing whether
two query vertices $u$ and $v$ are strongly connected in $G \setminus e$,
 reporting the SCCs of $G \setminus e$, or the largest and smallest SCCs in $G \setminus e$, for any query edge $e$.
 Furthermore, using our framework, it is possible to maintain the $2$-vertex-and $2$-edge-connected components of a digraph under edge deletions.
All of these extensions can be handled with the same time and space bounds as for the extended decremental SCC problem. (Most of these reductions have been deferred to the full version of the paper.)

A naive approach to solving the extended decremental SCC problem is to maintain separately the SCCs in every subgraph $G \setminus w$ of $G$, for all vertices $w$.
After an edge deletion we then update the SCCs of all these $n$ subgraphs, where $n$ is the number of vertices in $G$.
If we simply perform a static recompution after each deletion, then we, for example, obtain decremental algorithms with $O(m^2n)$ total time and $O(n^2)$ space by recomputing the SCCs in each $ G \setminus w $~\cite{Tarjan72} or $O(m^2+mn)$ total time and $O(m + n)$ space by constructing a more suitable static connectivity data structure~\cite{GIP:SC}, respectively.
Here $ m $ denotes the initial number of edges.
The current fastest (randomized) decremental SCC algorithm by Chechik et al.~\cite{ChechikHILP16} trivially gives $O(m n^{3/2} \log{n})$ total update time and $O(m n)$ space for our extended decremental SCC problem.

The main technical contribution of this paper is a data structure for the extended decremental SCC problem with $O(m n \log{n})$ total update time that uses $O(n^2 \log{n})$ space, and that answers queries in constant time.
We obtain this data structure by extending  \L\k{a}cki's decremental SCC algorithm~\cite{scc-decomposition}. His algorithm maintains the SCCs of a graph under edge deletions by recursively decomposing the SCCs into smaller and smaller subgraphs. We therefore refer to his data structure as an SCC-decomposition. His total update time is $O(mn)$ and the space used is $O(m+n)$. We observe that the naive algorithm based on SCC-decompositions can be implemented in such a way that most of the work performed is redundant. We obtain our data structure by merging $n$ SCC-decompositions into one joint data structure, which we refer to as a joint SCC-decomposition. Our data structure, like that of \L\k{a}cki, is deterministic.
Using completely different techniques, Georgiadis et al.~\cite{Inc2C} showed how 
to answer the same sensitivity queries in $O(mn)$ total time in the 
incremental setting, i.e., when the input digraph undergoes edge 
insertions only.

The extended SCC problem is related to 
the so-called \emph{fault-tolerant model}. Here, one wishes to preprocess a graph $G$ into a data structure that is able to answer fast
certain sensitivity queries, i.e., given a
failed vertex $w$ (resp., failed edge $e$), compute a specific property of the subgraph $G \setminus w$ (resp., $G \setminus e$) of $G$.
%
Our data structure
supports sensitivity queries when a digraph $G$ undergoes edge deletions, which gives an aspect of \emph{decremental fault-tolerance}.
This may be useful in scenarios where
we wish to find the best edge whose deletion optimizes certain properties (fault-tolerant aspect)
and then actually perform this deletion (decremental aspect).
This is, e.g., done in the computational biology applications considered by Mihal\'ak et al.~\cite{MihalakUY15}. Their recursive deletion-contraction
algorithm repeatedly finds the edge of a strongly connected digraph
whose deletion maximizes quantities such as the number of resulting SCCs or minimizes their maximum size.

%

As another important application of our joint SCCs data structure, we provide the first nontrivial algorithm for maintaining the dominator tree of a flow graph under edge deletions.
A flow graph $G=(V,E,s)$ is a directed graph with a distinguished start vertex $s \in V$, w.l.o.g.\ containing only vertices reachable from $ s $.
A vertex $w$ \emph{dominates} a vertex $v$ ($w$ is a \emph{dominator} of $v$) if every path from $s$ to $v$ includes $w$.
The \emph{immediate dominator} of a vertex $ v $, denoted by $ d(v) $, is the unique vertex that dominates $ v $ and is dominated by all dominators of $ v $.
The \emph{dominator tree} $ D $ is a tree with root $ s $ in which each vertex $ v $ has $ d(v) $ as its parent.
Dominator trees can be computed in linear time~\cite{domin:ahlt,dominators:bgkrtw,domin:bkrw,dom:gt04}.
The problem of finding dominators has been extensively studied, as it occurs in several applications, including
program optimization and code generation~\cite{cytron:91:toplas},
constraint programming~\cite{QVDR:PADL:2006},
circuit testing~\cite{amyeen:01:vlsitest}, theoretical biology~\cite{foodwebs:ab04}, memory profiling~\cite{memory-leaks:mgr2010},
fault-tolerant computing~\cite{FaultTolerantReachability,FaultTolerantReachability:STOC16},
connectivity and path-determination problems~\cite{2vc,2VCSS:Geo,2ECB,2VCB,2CC:HenzingerKL15,Italiano2012,2vcb:jaberi15,2VCC:Jaberi2016}, and the analysis of diffusion networks~\cite{Rodrigues:icml12}.
%
%

In particular, the dynamic dominator problem arises in various applications, such as data flow analysis and compilation~\cite{semidynamic-digraphs:CFNP,Gargi:2002}.
Moreover, the results of Italiano et al.~\cite{Italiano2012} imply that 
dynamic dominators
can be used for dynamically testing 2-vertex connectivity, and for maintaining the strong bridges and strong articulation points of digraphs. The decremental dominator problem appears in the computation of maximal 2-connected subgraphs in digraphs~\cite{2CC:HenzingerKL15,2VCC:Jaberi2016,2CExp}.
The problem of updating the dominator relation has been studied for a few decades (see, e.g.,~\cite{dynamicdominator:AL,incrementaldominators:CR,semidynamic-digraphs:CFNP,dyndom:2012,PGT11,RR:incdom,SGL}).
For the incremental dominator problem, there are algorithms that achieve total $O(mn)$ running time for processing a sequence of edge insertion in a flow graph with $n$ vertices, where $m$ is the number of edges after all insertions~\cite{dynamicdominator:AL,semidynamic-digraphs:CFNP,dyndom:2012}.
Moreover, they can answer dominance queries, i.e., whether a query vertex $w$ dominates another query vertex $v$, in constant time.
Prior to our work, to the best of our knowledge, no decremental algorithm with total running time better than $O(m^2)$ was known for general flow graphs.
In the special case of \emph{reducible} flow graphs (a class that includes acyclic flow graphs),
Cicerone et al.~\cite{semidynamic-digraphs:CFNP} achieved an $O(mn)$ update bound for the decremental dominator problem. Both the incremental and the decremental algorithms of \cite{semidynamic-digraphs:CFNP} require $O(n^2)$ space, as they maintain the transitive closure of the digraph.

Our algorithm is the first to improve the trivial $O(m^2)$ bound for the decremental dominator problem in \emph{general} flow graphs. Specifically, our algorithm can process a sequence of edge deletions in a flow graph with $n$ vertices and initially $m$ edges in $O(mn \log{n})$ time and $O(n^2 \log{n})$ space, and after processing each deletion can answer dominance queries in constant time.
For the special case of  \emph{reducible} flow graphs, we give an algorithm that matches the $O(mn)$ running time of Cicerone et al. while improving the space usage to $O(m+n)$.
We remark that the reducible case is interesting for applications in program optimization since one notion of a ``structured'' program is that its flow graph is reducible. (The details about this result appear in the full version of the paper.)
Finally, we complement our results with a conditional lower bound, which suggests that it will be hard to substantially improve our update bounds. In particular, we prove that there is no incremental nor decremental algorithm for maintaining the dominator tree (or more generally, a dominance data structure) that has total update time $ O ((m n)^{1-\epsilon}) $ (for some constant $ \epsilon > 0 $) unless the \textsf{OMv} Conjecture~\cite{HenzingerKNS15} fails. The same lower bound applies to the extended decremental SCC problem. Unlike the update time, it is not clear that the $O(n^2\log n)$ space used by our joint SCC-decomposition is near-optimal. We leave it as an open problem to improve this bound.

\section{Notation and Terminology}

For a given directed graph $G = (V,E)$, we denote the set of vertices by $V(G) = V$ and the set of edges by $E(G) = E$. We let $m$ and $n$ be the number of edges and vertices, respectively, of~$G$.
Two vertices $ u $ and $ v $ are \emph{strongly connected} in $ G $ if there is a path from $ u $ to $ v $ as well as a path from $ v $ to $ u $ in $ G $ and $G$ is \emph{strongly connected} if every vertex is reachable from every other vertex.
The \emph{strongly connected components} (SCCs) of $G$ are its maximal strongly connected subgraphs.
The SCCs of a graph can be computed in $O(m+n)$ time~\cite{Tarjan72}.
We denote by $G \setminus S$ (resp., $G \setminus (u, v)$) the graph obtained after 
deleting a set $S$ of vertices (resp., an edge $(u, v)$) from $G$.
Additionally, we let $G[S]$ be the subgraph of $G$ induced by the set of vertices $S$.
For a strongly connected graph $H$, we say that deleting an edge $(u, v)$ \emph{breaks} $H$, if $H \setminus (u, v)$ is not strongly connected.


An edge (resp., a vertex) of $G$ is a \emph{strong bridge} (resp., a \emph{strong articulation point}) if its removal increases the number of
SCCs. 
Let $G$ be a strongly connected graph. 
We say that $G$ is $2$-edge-connected (resp., $2$-vertex-connected) if it has no strong
bridges (resp., at least three vertices and no strong articulation points). 
For a set of vertices~$C \subseteq V$ its induced subgraph $G[C]$ is a {\em maximal $2$-edge-connected subgraph} (resp., {\em maximal $2$-vertex-connected subgraph}) of $G$ if 
$G[C]$ is a $2$-edge-connected (resp., $2$-vertex-connected) graph and no superset of $C$ has this property.

Two vertices $u$ and $w$ are $2$-edge-connected (resp., $2$-vertex-connected) if there are two edge-disjoint (resp., internally vertex-disjoint) paths from $u$ to $w$ and two edge-disjoint (resp., internally vertex-disjoint) paths from $w$ to $u$.
(Note that a path from $u$ to $w$ and a path from $w$ to $u$ need not be edge-disjoint or internally vertex-disjoint.)
A $2$-edge-connected (resp., $2$-vertex-connected) component of $G$ is a maximal subset of vertices such that any pair of distinct vertices is $2$-edge-connected (resp., $2$-vertex-connected). 

We denote by $ G^{\R} $ the \emph{reverse graph} of $ G $, i.e., the graph which has the same vertices as $ G $ and contains an edge $ (v, u) $ for every edge $ (u, v) $ of $ G $.
If $ D $ is the dominator tree of $ G $, then $ D^{\R} $ denotes the dominator tree of $ G^{\R} $.
A \emph{spanning tree} $T$ of $G$ is a tree with root $s$ that contains a path from $s$ to $v$ for all reachable vertices $v$.
Given a rooted tree $T$, we denote by $T(v)$ the subtree of $T$ rooted at $v$ (we also view $T(v)$ as the set of descendants of $v$).

Let $G=(V,E,s)$ be a flow graph with start vertex $s$, and let $D$ be the dominator tree of $G$.
We represent $G$ by adjacency lists $\mathit{In}(v) = \{ u : (u,v) \in E \}$ and $\mathit{Out}(v) = \{ w : (v,w) \in E \}$.
We represent $D$ by storing parent and child pointers, i.e., each vertex $v$ stores its parent $d(v)$ in $D$ and the list of children $C(v)$.
Let $T$ be a tree rooted at $s$ with vertex set $V(T) \subseteq V$, and let $t(v)$ denote the parent of a vertex $v \in V(T)$ in $T$.
If $v$ is an ancestor of $w$, $T[v, w]$ is the path in $T$ from $v$ to $w$.
In particular, $D[s,v]$ consists of the vertices that dominate $v$.
If $v$ is a proper ancestor of $w$, $T(v, w]$ is the path to $w$ from the child of $v$ that is an ancestor of $w$. Tree $T$ is \emph{flat} if its root is the parent of every other vertex.
For any vertex $v \in V$, we denote by $C(v)$ the set of children of $v$ in $D$.

\section{A Data Structure for Maintaining Joint SCC-Decompositions}\label{sec:joint}

For a given initial graph $G$, the \emph{decremental SCC problem} asks us to maintain a data structure that allows 
edge deletions and 
can answer whether (arbitrary) pairs $(u,v)$ of vertices are in the same SCC. The goal is to update the data structure as quickly as possible while answering queries in constant time. In this paper we present a data structure for the \emph{extended decremental SCC problem} in which a query provides an additional vertex $w$ and asks whether $u$ and $v$ are in the same SCC when $w$ is deleted from $G$. We 
maintain this information under edge deletions, and our data structure relies on Łącki's \emph{SCC-decomposition}~\cite{scc-decomposition} for doing so.

\subsection{Review of Łącki's SCC Decomposition}

An SCC-decomposition recursively partitions the graph $G$ into smaller strongly connected subgraphs.
This generates a
rooted tree $T$, whose root $r$ represents the entire graph, and
where the subtree rooted at each node $\phi$ represents some vertex-induced strongly connected subgraph
$G_\phi$ (we refer to vertices of $T$ as nodes to distinguish $T$ from $G$).
Every non-leaf node $\phi$ is a vertex of $G_\phi$, and the children of $\phi$ correspond to SCCs of $G_\phi \setminus \phi$. The concept was introduced by Łącki~\cite{scc-decomposition} and was slightly extended by Chechik et al.~\cite{ChechikHILP16} to allow \emph{partial} SCC-decompositions where leaves represent strongly connected
subgraphs rather than single vertices. We adopt the notation from \cite{ChechikHILP16}.

\begin{definition}[SCC-decomposition]\label{def:scc-decomposition}
Let $G=(V, E)$ be a strongly connected graph. An \emph{SCC-decomposition}
of $G$ is a rooted tree $T$, whose nodes form a partition of $V$.
For a node $\phi$ of $T$ we define $G_{\phi}$ to be the subgraph of $G$ induced by the union of all descendants of $\phi$ (including $\phi$).
Then, the following properties hold:
\begin{itemize}
\item
Each internal node $\phi$ of $T$ is a single-element set.\footnote{In this case, we sometimes abuse notation and assume that $\phi$ is the vertex itself.}
\item
Let $\phi$ be any internal node of $T$, and let $H_1,\dots,H_t$ be the
SCCs of $G_\phi \setminus \phi$. Then the node $\phi$ has $t$ children
$\phi_1,\dots,\phi_t$, where $G_{\phi_i} = H_i$ for all $i \in \{1,\dots,t\}$.
\end{itemize}
An SCC-decomposition of a graph $G$ that is not strongly connected is
a collection of SCC-decompositions of the SCCs of $G$. We say that $T$ is
a \emph{partial} SCC-decomposition when the leaves of $T$ are not
required to be singletons.
\end{definition}

Observe that for each node $\phi$, the graph $G_{\phi}$ is strongly connected.
Moreover, the subtree of~$T$ rooted at $\phi$ is an SCC-decomposition of $G_\phi$. Also, for a leaf $\phi$ we have that $\phi = V(G_\phi)$.
To build an SCC-decomposition $T$ of a strongly connected
graph $G$ we pick an arbitrary vertex $v$, put it in the root of $T$,
then recursively build SCC-decompositions of SCCs of $G \setminus \{v\}$
and make them the children of $v$ in $T$. This procedure is described in $\textsc{Build-SCC-Decomposition}(G,S)$.
Note that since the choice of $v$ is arbitrary, there are many ways to build an SCC-decomposition of the
same graph. The procedure $\textsc{Build-SCC-Decomposition}(G,S)$
takes as input a set of vertices $S$ and returns a partial
SCC-decomposition whose internal nodes are the vertices of $S$, i.e.,
these vertices are picked first and therefore appear at the top of the
constructed tree.
We refer to the vertices in~$S$ as internal nodes and the
remaining nodes as external nodes.
Note that all external nodes appear in the leaves of $T$, while internal nodes can be both leaves and non-leaves. This distinction is helpful when describing our algorithm.
We therefore let
$\textsc{Internal}(T)$ be the nodes of $T$ from $S$ and
$\textsc{External}(T)$ be the nodes of $T$ that are not from $S$. In
particular, $\textsc{External}(T)$ is a subset of the leaves of $T$.

\begin{figure}
\begin{minipage}{\linewidth}
  \begin{procedure}[H]
  \DontPrintSemicolon
  \KwIn{A strongly connected graph $G$ and $S \subseteq V(G)$.}
  \KwOut{A partial SCC-decomposition $T$ of $G$ whose internal
  nodes are the vertices of $S$.}
  \BlankLine
  \If{$S = \emptyset$}{\Return the tree $T$ consisting of a single node $\phi = V(G)$ with associated graph $G_\phi = G$.\;}

  Pick an arbitrary vertex $v \in S$.\;

  Make $v$ the root of $T$, and let $G_v = G$.\;

Compute the SCCs $H_1,\dots,H_t$ of $G \setminus \{v\}$.\;

\ForEach{$i \in \{1,\dots,t\}$}{
  Recursively compute $T_i = \builddec(H_i, S \cap V(H_i))$.\;
  Make the subtree $T_i$ a child of $v$ in $T$.\;
  }

\Return $T$.\;

  \caption{Build-SCC-Decomposition($G$,$S$)}
  \label{proc:Build-SCC}
\end{procedure}
\end{minipage}
\end{figure}

Łącki~\cite{scc-decomposition} showed that the total initialization and update time under edge deletions of an SCC-decomposition
is $O(m\gamma)$, where $\gamma$ is the depth of the decomposition.

\subsection{Towards a Joint SCC-Decomposition}\label{sec:joint initialization}

Recall that the extended decremental SCC problem asks us to maintain under edge deletions a data structure for a graph $G$ such that we can answer whether $u$ and $v$ are strongly connected in $G \setminus \{w\}$ when given $u,v,w \in V(G)$. A naive algorithm does this by maintaining $n$ SCC-decompositions, each with a distinct vertex $w$ as its root. The children of $w$ in an SCC-decomposition that has $w$ as its root are then exactly the SCCs of $G \setminus \{w\}$. Hence, $u$ and $v$ are in the same SCC if and only if they appear in the same subtree below $w$. The total update time of this data structure is however $O(mn^2)$, which is undesirable.
With a more refined approach, we improve the time bound to $O(mn\log n)$.

Observe that the external nodes of a partial
SCC-decomposition $T$ produced by the procedure
$\textsc{Build-SCC-Decomposition}(G,S)$ exactly correspond to the SCCs
of $G \setminus S$. This is true regardless of the order in which
vertices from $S$ are picked by the procedure. If two
SCC-decompositions are built using the same set $S$, but with vertices
being picked in a different order, then the nodes below $S$
represent the same SCCs, which means that they can be shared by the
two SCC-decompositions. Our algorithm is based on this observation. We
essentially construct the $n$ SCC-decompositions of the naive
algorithm described above such that large parts of their subtrees are
shared, and such that we do not need to maintain multiple copies of
these subtrees. The idea is to partition the set $S$ into two subsets $S_1$ and $S_2$ of equal size (we assume for simplicity that $n$ is a power of 2), and then construct half of the
SCC-decompositions with $S_1$ at the top and the other half with $S_2$ at
the top. The procedure is repeated recursively on the top part of both halves. We refer to the bottom part, i.e., nodes that are not from $S_1$ and $S_2$, respectively, as the
\emph{extension} of the top part.
Note that we eventually get a distinct vertex as the root of each of the $n$
SCC-decompositions. The following definition formalizes the idea.

\begin{definition}[Joint SCC-decomposition]
A \emph{joint SCC-decomposition} $J$ is a recursive structure. It is
either a regular SCC-decomposition $T$ (the base case), or a pair of
joint SCC-decompositions $J_1,J_2$ with the same set of internal nodes
$S$ and a shared set of external nodes $\Phi$. 
In the second case we refer to $J$ as the tuple $(J_1,J_2,S,\Phi)$.
A joint SCC-decomposition $J = (J_1,J_2,S,\Phi)$ is \emph{balanced} on $S$ if it has one of the following two properties:
\begin{enumerate}
\item
$S$ is a singleton and $J$ is a regular (partial) SCC-decomposition $T$ with the vertex from $S$ as root and no other internal nodes (the base case).
\item
$S$ can be partitioned into two equally sized halves $S_1$ and $S_2$, and 
  $J$ consists of two joint SCC-decompositions $J_1 =
  (J_{1,1},J_{1,2},S_1,\Phi_1)$ and $J_2 =
  (J_{2,1},J_{2,2},S_2,\Phi_2)$ that are balanced on $S_1$ and $S_2$, respectively. Also, each external node $\phi$ in
  $\Phi_1$ and $\Phi_2$ is extended with an associated
  SCC-decomposition $T_\phi$ for $G_\phi$ whose internal nodes are
  those of $\phi \cap S$. The combined set of external nodes of $T_\phi$ for
  all $\phi \in \Phi_1$ is equal to the combined set of external nodes of $T_{\phi'}$ for
  all $\phi' \in \Phi_2$, and these nodes are the external nodes
  $\Phi$ of $J$.
\end{enumerate}
\end{definition}

The procedure $\textsc{Build-Joint-SCC-Decomposition}(G,S)$ describes
how we build a balanced joint SCC-decomposition. $G$ is the graph that
we wish to decompose, and $S$ is the set of vertices that we wish to
place at the top. Initially $S$ is the set of all vertices. If $S$
only contains a single vertex $r$, then we make $r$ the root of a
regular SCC-decomposition. Note that in this case the vertex $r$ is
the only internal node of the partial SCC-decomposition returned by
$\textsc{Build-SCC-Decomposition}(G,S)$. If $S$ contains more than one
vertex, then we split it into two equal halves $S_1$ and $S_2$ and
recursively compute a joint SCC-decomposition for each half. The procedure
$\textsc{Build-Joint-SCC-Decomposition}(G,S_i)$ only uses vertices
from $S_i$ as internal nodes, and we therefore compute regular
SCC-decompositions for the remaining vertices from~$S$ for each of the
resulting external nodes. This gives us two structures that both have the
vertices from $S$ as internal nodes, and since their external nodes are shared
they form a joint SCC-decomposition. We add the external nodes to a list $\Phi$ that is used as an interface between the different SCC-decompositions. Recall that each node is a subset of the vertices in $G$, and observe that the nodes in $\Phi$ form a partition of $V(G) \setminus S$.

\begin{figure}
\begin{minipage}{\linewidth}
\begin{procedure}[H]
  \DontPrintSemicolon
  \KwIn{A graph $G$ and a set of vertices $S \subseteq V(G)$}
  \KwOut{A balanced joint SCC-decomposition $J = (J_1,J_2,S,\Phi)$ of $G$ on $S$.}
  \BlankLine
  \If{$|S| = 1$}{
    \Return $T = \textsc{Build-SCC-Decomposition}(G,S)$. \label{alg:init}\;
  }

  Let $S_1$ and $S_2$ be the first and second half of $S$, respectively,
  and let $\Phi$ be an empty list.\;
  \ForEach{$i \in \{1,2\}$}{
    Compute $J_i = \textsc{Build-Joint-SCC-Decomposition}(G,S_i)$.\;
    \ForEach{external node $\phi$ of $J_i$}{
      Compute $T_\phi = \textsc{Build-SCC-Decomposition}(G_\phi,\phi
      \cap S)$.\;
      Add each external node of $T_\phi$ to $\Phi$, if it is not already there.\;\label{line:merge}
    }
  }

  \Return $J = (J_1,J_2,S,\Phi)$
  
  \caption{Build-Joint-SCC-Decomposition($G$,$S$)}
  \label{proc:Build-Joint-SCC}
\end{procedure}
\end{minipage}
\end{figure}

\begin{lemma}\label{lemma:num-scc-decompositions}
A balanced joint SCC-decomposition for a graph $G$ with $n$ vertices consists of $O(n)$ SCC-decompositions.
\end{lemma}

\begin{proof}
  Observe that the procedure $\textsc{Build-Joint-SCC-Decomposition}(G,S)$ constructs a number of SCC-decompositions that is given by the following simple recurrence, where $|E(G)| = m$ and $|S| = s$:
\[
g(m,s) = \begin{cases}
  2g(m,s/2) + 2 & \text{if $s > 1$}\\
  1 & \text{otherwise}
\end{cases}
\]
Since $g(m,s) = O(s)$, the lemma follows.
\end{proof}

The following lemma shows that a joint SCC-decomposition achieves a
much more compact representation than the naive algorithm described at
the beginning of the section. We will later use the lemma in our analysis.

\begin{lemma}\label{lemma:count}
Let $J = (J_1,J_2,S,\Phi)$ be a balanced joint SCC-decomposition of a graph $G$ such that $S = V(G)$. Then the total number of nodes of $J$ is $O(n\log n)$, where $n = |V(G)|$.
\end{lemma}

\begin{proof}
The proof is by induction. Our induction hypothesis says that
the total number of internal nodes of a balanced joint
SCC-decomposition $J = (J_1,J_2,S,\Phi)$, counting not only $S$ but
also recursively the number of internal nodes of $J_1$ and $J_2$, is
$|S| \cdot (1+\log |S|)$. 

In the base case, $J = (J_1,J_2,S,\Phi)$ is an SCC-decomposition with a single internal
node, and the induction hypothesis is clearly satisfied. For the
induction step we count separately the total number of internal nodes of
$J_1 = (J_{1,1},J_{1,2},S_1,\Phi_1)$ and $J_2 =
(J_{2,1},J_{2,2},S_2,\Phi_2)$, and add the number of internal nodes
of the SCC-decompositions $T_\phi$ for $\phi \in \Phi_1$ and $\phi \in
\Phi_2$, i.e., the extensions of $J_1$ and $J_2$ to $S$. Since $|S_1|
= |S_2| = |S|/2$, it follows from the induction hypothesis that both
$J_1$ and $J_2$ have $\frac{|S|}{2} \cdot \log |S|$ internal nodes in
total. The internal nodes of $T_\phi$ for $\phi \in \Phi_1$ are
exactly $S_2$, and the internal nodes of $T_\phi$ for $\phi \in \Phi_2$ are
exactly $S_1$. Hence the number of internal nodes in the extensions are
$|S_1|+|S_2| = |S|$. It follows that the total number of internal
nodes of $J$ is $|S| \cdot (1+\log |S|)$ as desired.

It remains to count the external nodes of $J$. Note that external nodes of $J_1$ and $J_2$ correspond to internal nodes of $J$, i.e., they are roots of the SCC-decompositions that extend $J_1$ and $J_2$. Therefore there are at most as many external nodes inside the recursion as there are internal nodes in total. There are at most $O(n)$ external nodes in the extensions of $J_1$ and $J_2$ to $S$, and we therefore conclude that the total number of nodes when $S=V(G)$ is
at most $O(n\log n)$.
\end{proof}

\begin{lemma}\label{lemma:init}
The procedure $\textsc{Build-Joint-SCC-Decomposition}(G,S)$ constructs a joint SCC-decomposition in time $O(mn\log n)$.
\end{lemma}

\begin{proof}
  We already argued that $\textsc{Build-Joint-SCC-Decomposition}(G,S)$ constructs a joint SCC-decomposition with vertices from $S$ as internal nodes. It thus remains to show that this is done in $O(mn\log n)$ time.

  Recall that Łącki~\cite{scc-decomposition} showed that an SCC-decomposition of depth $\gamma$ can be initialized in time $O(m\gamma)$. This follows straightforwardly from the fact that the SCCs of a graph can be computed in $O(m+n)$ time~\cite{Tarjan72}. The depth of an SCC-decomposition is at most equal to the number of internal nodes plus one. Since $\textsc{Build-SCC-Decomposition}(G,S)$ only uses vertices from $S$ as internal nodes, it runs in time $O(|E(G)| \cdot (|S|+1))$. The running time of $\textsc{Build-Joint-SCC-Decomposition}(G,S)$ is dominated by the calls that it makes to the procedure $\textsc{Build-SCC-Decomposition}(G,S)$. In particular, it is not difficult to merge the leaves on line \ref{line:merge} in linear time. It follows that the running time of $\textsc{Build-Joint-SCC-Decomposition}(G,S)$ when $|E(G)| = m$ and $|S| = s$ is upper bounded by the following recurrence for some constant $c$:
\[
f_c(m,s) = \begin{cases}
  2f_c(m,s/2) + c m s & \text{if $s > 1$}\\
  c\cdot m & \text{otherwise}
\end{cases}
\]
Since $f_c(m,s) = \sum_{i=0}^{\log s} cms = cms\log s$, the claim follows.  

Alternatively, the lemma can be proved by observing that Lemma~\ref{lemma:count} shows that the total number of nodes in a joint SCC-decomposition is $O(n\log n)$, which implies that the combined depth of all the SCC-decompositions that make up a joint SCC-decomposition is at most $O(n\log n)$.
\end{proof}

To answer queries for the extended decremental SCC problem in constant time, we also construct and maintain an $n \times n$ matrix $A$ such that $A[u,w]$ is the index of the SCC of $G \setminus \{w\}$ that contains $u$. Two vertices $u$ and $v$ are in the same SCC of $G \setminus \{w\}$ if and only if $A[u,w] = A[v,w]$. To avoid cluttering the pseudo-code we describe separately how $A$ is maintained. In $\textsc{Build-Joint-SCC-Decomposition}(G,S)$ we initialize $A$ in the base case when we compute an SCC-decomposition $T$ for a singleton $S = \{w\}$. Indeed, in this case $w$ is the root of $T$, and the external nodes are exactly the SCCs of $G\setminus \{w\}$. Hence, for every vertex $u \in V(G) \setminus \{w\}$ we make $A[u,w]$ equal to the index of the SCC it is part of in $G\setminus \{w\}$.

Note that storing the matrix $A$ takes space $O(n^2)$. The time spent initializing $A$ is however dominated by the other work performed by the algorithm.

\subsection{Deleting Edges from a Joint SCC-Decomposition}\label{sec:delete}

We next show how to maintain a joint SCC-decomposition under edge deletions. It is again instructive to consider the work performed by the naive algorithm that maintains $n$ SCC-decompositions with distinct roots. If these are constructed as described 
in Section~\ref{sec:joint initialization}, then the SCC-decompositions will share many identical subtrees, and the work performed on these subtrees will be the same. In the joint SCC-decomposition such subtrees are shared, but otherwise the work performed is the same as the work performed for individual SCC-decompositions. We therefore use Łącki's algorithm~\cite{scc-decomposition} to delete edges from the individual SCC-decompositions, and we introduce a new procedure for handling the interface between the SCC-decompositions. We next briefly sketch Łącki's algorithm. We refer to \cite{ChechikHILP16,scc-decomposition} for a more comprehensive presentation.

Recall that each node $\phi$ of an SCC-decomposition $T$ represents a strongly connected subgraph $G_\phi$ induced by the vertices in the subtree rooted at $\phi$. If $\phi$ is an internal node of $T$, then the children of $\phi$ are the SCCs of $G_\phi \setminus \phi$. Łącki uses the following two operations to compactly represent edges among $\phi$ and its children.

\begin{definition}
Let $G$ be a graph. The \emph{condensation} of $G$, denoted by $\condense(G)$, is the graph obtained from $G$ by contracting all its SCCs into single vertices.
Let $v \in V(G)$.
By $\splitop(G, v)$ we denote the graph obtained from $G$ by splitting $v$ into two vertices: $v_{in}$ and $v_{out}$.
The in-edges of $v$ are connected to $v_{in}$ and the out-edges to $v_{out}$.
\end{definition}

The two operations are often used together, and to simplify notation we use the shorthand $G^{\con}_v = \condense(\splitop(G, v))$. The graph $G^{\con}_\phi$ is stored with every internal node $\phi$ of the SCC-decomposition $T$. This introduces at most three copies of every vertex $v$ of $G$: The two vertices $v_{in}$ and $v_{out}$ in $G^{\con}_v$, and possibly a third vertex in the condensed graph of the parent of $v$ in $T$. Moreover, every edge $(u, v)$ appears in exactly one condensed graph, namely that of the lowest common ancestor of $u$ and $v$ in $T$, which we denote by $\LCA(u,v)$. The combined space used for storing all the condensed graphs is thus $O(m+n)$.

To delete an edge $(u, v)$, Łącki~\cite{scc-decomposition} locates $\phi = \LCA(u,v)$, and deletes $(u', v')$ from $G^{\con}_\phi$, where $u'$ and $v'$ are the vertices whose subtrees contain $u$ and $v$. (He uses $O(m)$ space to store a pointer from every edge $(u, v)$ to $\LCA(u,v)$, enabling him to find the lowest common ancestor in constant time.) To preserve connectivity, he then checks whether $u'$ and $v'$ have non-zero out- and in-degrees, respectively, in $G^{\con}_\phi$. If this is not the case, then he repeatedly removes vertices with out- or in-degree zero and their adjacent edges from $G^{\con}_\phi$. All such vertices can be located, starting from $u'$ and $v'$, in time that is linear in the number of edges adjacent to the removed vertices. The corresponding children of $\phi$ are then moved up one level in $T$ and are made siblings of $\phi$. They are also inserted into $G^{\con}_{par(\phi)}$, where $par(\phi)$ is the parent of $\phi$, and their edges and the edges of $\phi$ in $G^{\con}_{par(\phi)}$ are updated correspondingly. This can again be done in time linear in the number of edges in the original graph that are adjacent to vertices in the subtrees that are moved. The procedure is then repeated in $G^{\con}_{par(\phi)}$. Since every vertex increases its level at most $\gamma$ times, where $\gamma$ is the initial depth of $T$, it follows that the total update time of the algorithm is at most $O(m\gamma)$.

We let $\textsc{Delete-Edge-from-SCC-decomposition}(T,u,v)$ be the procedure described above for deleting an edge $(u, v)$ from an SCC-decomposition $T$. We also denote the recursive procedure for moving nodes $\phi_1,\dots,\phi_k$ from being children of $\phi$ to being siblings of $\phi$ in $T$ after an edge $(u, v)$ is deleted by $\textsc{Fix-SCC-decomposition}(T,u,v,\phi,\{\phi_1,\dots,\phi_k\})$. Both procedures return the resulting SCC-decomposition, or a collection of SCC-decompositions in case the graph is not strongly connected. Armed with these two procedures, we are now ready to describe how our algorithm updates a joint SCC-decomposition when an edge $(u, v)$ is deleted.

\begin{procedure}[t!]
  \DontPrintSemicolon
  \KwIn{An SCC-decomposition $T$ for a strongly connected graph $G$, and an edge $(u, v)$ to be deleted from $G$.}
  \KwOut{The updated collection of SCC-decompositions $\{T_1,\dots,T_k\}$ for $G\setminus (u, v)$.}
  \BlankLine

  Let $\phi = \LCA(u,v)$ be the lowest common ancestor of $u$ and $v$ in $T$.\;
  Let $u'$ and $v'$ be the vertices of $G^{\con}_\phi$ that are ancestors of $u$ and $v$, respectively, in $T$.\;
  Remove one copy of $(u', v')$ from $G^{\con}_\phi$.\;

  Starting from $u'$, find all vertices $\phi'_1,\dots,\phi'_r$ in $G^{\con}_\phi$ that cannot reach $\phi_{in}$.
  Starting from $v'$, find all vertices $\phi''_1,\dots,\phi''_s$ in $G^{\con}_\phi$ that are unreachable from $\phi_{out}$.
  Remove $\phi'_1,\dots,\phi'_r,\phi''_1,\dots,\phi''_s$ and their adjacent edges from $G^{\con}_\phi$, and remove the corresponding subtrees from $T$.

  \Return \textsc{Fix-SCC-decomposition}($T$,$u$,$v$,$\phi$,$\{\phi'_1,\dots,\phi'_r,\phi''_1,\dots,\phi''_s\}$).\;

  \caption{Delete-Edge-from-SCC-decomposition($T$,$u$,$v$)}
  \label{proc:delete-scc}
\end{procedure}

\begin{procedure}[t!]
  \DontPrintSemicolon
  \KwIn{A broken SCC-decomposition $T$ for a graph $G$. The two endpoints $u$ and $v$ of the edge whose removal broke $T$. A node $\phi$ of $T$, and a set $\{\phi_1,\dots,\phi_k\}$ of nodes that are to be made siblings of $\phi$.}
  \KwOut{A collection of valid SCC-decompositions $\{T_1,\dots,T_\ell\}$ for $G \setminus (u, v)$.}
  \BlankLine

  \If{$\phi$ has no parent, or $k = 0$}{
    \Return $\{T,T_{\phi_1},\dots,T_{\phi_\ell}\}$.\;
  }

  Let $\phi'$ be the parent of $\phi$ in $T$.\;

  Add $\phi_1,\dots,\phi_k$ as children of $\phi'$ in $T$, and add $\phi_1,\dots,\phi_k$ to $G^{\con}_{\phi'}$.\;

  Let $E'$ be the edges of $G_{\phi'}$ for which one end-point is part of $G_{\phi_i}$, for some $i \in \{1,\dots,k\}$, and the other end-point is not part of $G_{\phi_i}$.\;

  \ForEach{$(u', v') \in E'$}{
    \eIf{$u' = \phi'$}{
      Let $u'' = \phi'_{out}$.\;
    }{
      Let $u''$ be the vertex of $G^{\con}_{\phi'}$ whose subtree contains $u'$.\;
    }
    \eIf{$v' = \phi'$}{
      Let $v'' = \phi'_{in}$.\;
    }{
      Let $v''$ be the vertex of $G^{\con}_{\phi'}$ whose subtree contains $v'$.\;
    }
    Add $(u'', v'')$ to $G^{\con}_{\phi'}$.\;
  }

  Starting from the vertex in $G^{\con}_{\phi'}$ that contains $u$, find all vertices $\phi'_1,\dots,\phi'_r$ in $G^{\con}_{\phi'}$ that cannot reach $\phi'_{in}$.
  Starting from the vertex in $G^{\con}_{\phi'}$ that contains $v$, find all vertices $\phi''_1,\dots,\phi''_s$ in $G^{\con}_{\phi'}$ that are unreachable from $\phi'_{out}$.
  Remove $\phi'_1,\dots,\phi'_r,\phi''_1,\dots,\phi''_s$ and their adjacent edges from $G^{\con}_{\phi'}$, and remove the corresponding subtrees from $T$.
  
  \Return \textsc{Fix-SCC-decomposition}($T$,$u$,$v$,$\phi'$,$\{\phi'_1,\dots,\phi'_r,\phi''_1,\dots,\phi''_s\}$).\;  

  \caption{Fix-SCC-decomposition($T$,$u$,$v$,$\phi$,$\{\phi_1,\dots,\phi_k\}$)}
  \label{proc:fix}
\end{procedure}

In a joint SCC-decomposition, vertices and edges may appear in multiple nodes as part of smaller SCC-decompositions. We therefore need to find every occurence of the edge that we wish to delete. We introduce a procedure $\textsc{Delete-Edge}(J,u,v)$ that does that by recursively searching through the nested joint SCC-decompositions and deleting $(u,v)$ from the relevant SCC-decompositions. The procedure also handles the interface between SCC-decompositions. Note that deleting $(u,v)$ from an SCC-decomposition $T$ may cause the SCC corresponding to the root $\phi$ of $T$ to break. The procedure $\textsc{Delete-Edge-from-SCC-decomposition}(T,u,v)$ will in this case return a collection of SCC-decompositions $\{T_1,\dots,T_k\}$, one for each new SCC. Suppose $J = (J_1,J_2,S,\Phi)$. If $T$ extends $J_1$ (resp. $J_2$), then it is an SCC-decomposition of the subgraph $G_\phi$ associated with some external node $\phi$ of $J_1$ (resp. $J_2$). $\phi$ is then itself a leaf of an SCC-decomposition $T'$ in $J_1$ (resp. $J_2$). Moreover, when the SCC corresponding to $\phi$ breaks, then this leaf must be split into multiple leaves of $T'$, one for each new SCC. Note however that the levels in $T'$ of the involved vertices do not change after the split. We therefore cannot charge the work performed when splitting $\phi$ to the analysis by Łącki~\cite{scc-decomposition}.

Let $\phi_1,\dots,\phi_k$ be the roots of the SCC-decompositions $T_1,\dots,T_k$ that are created when the deletion of $(u,v)$ breaks the SCC $G_\phi$. As mentioned above, we need to replace $\phi$ in $T'$ by $\phi_1,\dots,\phi_k$, which means that $\phi_1,\dots,\phi_k$ should replace $\phi$ in $G^{\con}_{par(\phi)}$, where $par(\phi)$ is the parent of $\phi$ in $T'$. 
To efficiently reconnect $\phi_1,\dots,\phi_k$ in $G^{\con}_{par(\phi)}$ we identify the vertex $\phi_i$ whose associated graph $G_{\phi_i}$ has the most vertices, and we then scan through all the vertices in the other graphs $G_{\phi_1},\dots,G_{\phi_{i-1}},G_{\phi_{i+1}},\dots,G_{\phi_{k}}$ and reconnect their adjacent edges in $G^{\con}_{par(\phi)}$ when relevant. The work performed is exactly the same as when Łącki fixes an SCC-decomposition after an edge is removed. We can therefore call \textsc{Fix-SCC-decomposition}($T'$,$u$,$v$,$\phi$,$\{\phi_1,\dots,\phi_{i-1},\phi_{i+1},\dots,\phi_k\}$). Note that this makes $\phi_i$ take over the role of $\phi$. Also note that we provide the procedure with the end-points $u$ and $v$ of the edge that was deleted, since $u$ and $v$ are used as starting points for the search for disconnected vertices when propagating the update further up the tree.

Finally, observe that splitting the leaf $\phi$ of the SCC-decomposition $T'$ may propagate all the way to the root of $T'$ and break the SCC corresponding to $T'$. We therefore use a recursive procedure, \textsc{Split-Leaf}($J$,$u$,$v$,$\phi$,$\{\phi_1,\ldots,\phi_k\}$), to perform the split. Here $u$ and $v$ are the end-points of the edge that was deleted and caused the need for the joint SCC-decomposition to be updated. We include them in the function call since they are used by $\textsc{Fix-SCC-decomposition}$ to initiate the search for vertices that are no longer strongly connected to ancestors of $\phi$.

\begin{procedure}[t!]
  \DontPrintSemicolon
  \KwIn{A balanced joint SCC-decomposition $J = (J_1,J_2,S,\Phi)$ for a graph $G$, and an edge $(u, v)$ to be deleted from $G$.}
  \KwResult{The edge $(u, v)$ is removed from $G$, and the joint SCC-decomposition $J$ is updated correspondingly.}
  \BlankLine

  \eIf{$\{u,v\} \subseteq V(G_\phi)$ for some $\phi \in \Phi$ \label{alg:pair}}{
    Let $T_\phi$ be the SCC-decomposition for $G_\phi$.\;
    $\{T_1,\ldots,T_k\} = \textsc{Delete-Edge-from-SCC-decomposition}(T_\phi,u,v)$.\;
    Let $\phi_1,\dots,\phi_k$ be the roots of $T_1,\ldots,T_k$.\;
    \If{$k > 1$}{
      $\textsc{Split-Leaf}(J_1,u,v,\phi,\{\phi_1,\ldots,\phi_k\})$.\;
      $\textsc{Split-Leaf}(J_2,u,v,\phi,\{\phi_1,\ldots,\phi_k\})$.\;
      Replace $\phi$ by $\phi_1,\dots,\phi_k$ in $\Phi$.\;
    }
  }{
    $\textsc{Delete-Edge}(J_1,u,v)$.\;
    $\textsc{Delete-Edge}(J_2,u,v)$.\;
  }

  \Return $J$.\;
  
  \caption{Delete-Edge($J$,$u$,$v$)}
  \label{proc:delete}
\end{procedure}

\begin{procedure}[t!]
  \DontPrintSemicolon
  \KwIn{A balanced joint SCC-decomposition $J = (J_1,J_2,S,\Phi')$ for a graph $G$. Two vertices $u$ and $v$, a node $\phi$ in the extension of $J$, and a collection of nodes $\{\phi_1,\ldots,\phi_k\}$. $u$ and $v$ are the end-points of an edge whose deletion causes $\phi_1,\ldots,\phi_k$ to break off from $\phi$.}
  \KwResult{The node $\phi$ is split by adding $\phi_1,\ldots,\phi_k$ to $J$, and $J$ is updated to correctly reflect the deletion of $(u, v)$.}
  \BlankLine

  Let $\phi' \in \Phi'$ be the external node of $J$ for which the SCC-decomposition $T_{\phi'}$ contains $\phi$ as a leaf.\;
  Let $i = \argmax_{j \in \{1,\dots,k\}} |V(G_{\phi_j})|$ be the index of the node $\phi_i$ whose associated graph $G_{\phi_i}$ contains the most vertices.\;

  \If{$\phi$ has a parent $par(\phi)$}{
    Remove from $G^{\con}_{par(\phi)}$ all edges that are adjacent to vertices from $G_{\phi_1},\dots,G_{\phi_{i-1}},G_{\phi_{i+1}},\dots,G_{\phi_k}$.\;
  }
  
  $\{T'_1,\dots,T'_\ell\} = \textsc{Fix-SCC-decomposition}(T_{\phi'},u,v,\phi,\{\phi_1,\dots,\phi_{i-1},\phi_{i+1},\dots,\phi_k\})$.\;
  \If{$\ell > 1$}{
    Let $\phi'_1,\dots,\phi'_\ell$ be the roots of $T'_1,\ldots,T'_\ell$.\;
    $\textsc{Split-Leaf}(J_1,u,v,\phi',\{\phi'_1,\ldots,\phi'_\ell\})$.\;
    $\textsc{Split-Leaf}(J_2,u,v,\phi',\{\phi'_1,\ldots,\phi'_\ell\})$.\;
    Replace $\phi'$ by $\phi'_1,\dots,\phi'_\ell$ in $\Phi'$.\;
  }
  
  \caption{Split-Leaf($J$,$u$,$v$,$\phi$,$\{\phi_1,\ldots,\phi_k\}$)}
  \label{proc:split}
\end{procedure}

Before analyzing the running time of our algorithm, we first describe a small implementation detail that was left out in the pseudo-code. Recall from Lemma~\ref{lemma:num-scc-decompositions} that a balanced joint SCC-decomposition for at graph $G$ with $n$ vertices is consists of $O(n)$ SCC-decompositions. For each SCC-decomposition $T$ we maintain an array on the vertices of the original graph $G$, such that $T\langle v\rangle = \textsc{True}$ if $v$ appears in $T$, and $T\langle v\rangle = \textsc{False}$ otherwise. This allows us to check in constant time whether $\{u,v\} \subseteq V(T_\phi)$ on line~\ref{alg:pair} of $\textsc{Delete-Edge}(J,u,v)$. Storing these arrays takes up $O(n^2)$ space, and they are updated when the SCC of the root of an SCC-decomposition breaks.

\begin{theorem}\label{thm:mn}
The total update time spent by $\textsc{Delete-Edge}(J,u,v)$ in order to maintain a balanced joint SCC-decomposition under edge deletions is $O(mn\log n)$.
\end{theorem}

\begin{proof}
  The time spent by $\textsc{Delete-Edge}(J,u,v)$ consists of three parts:
  \begin{enumerate}
  \item
    Checking whether $\{u,v\}\subseteq V(G_\phi)$ for some $\phi \in \Phi$ on line~\ref{alg:pair}.
  \item
    The work performed by $\textsc{Delete-Edge-from-SCC-decomposition}(T_\phi,u,v)$.
  \item
    The work performed by $\textsc{Split-Leaf}(J_i,u,v,\phi,\{\phi_1,\ldots,\phi_k\})$, for $i \in \{1,2\}$.
  \end{enumerate}

  As described above, we can check in constant time whether $\{u,v\}\subseteq V(G_\phi)$ by checking whether $T_\phi\langle u\rangle = \textsc{True}$ and $T_\phi\langle v\rangle = \textsc{True}$.

  Recall that Łącki~\cite{scc-decomposition} showed that the total initialization and update time of an SCC-decomposition is $O(m\gamma)$, where $\gamma$ is the depth of the decomposition. It follows from an argument identical to the proof of Lemma~\ref{lemma:init} that $\textsc{Delete-Edge-from-SCC-decomposition}(T_\phi,u,v)$ spends $O(mn\log n)$ time in total. Alternatively, Lemma~\ref{lemma:count} shows that the total number of nodes of the SCC-decompositions is $O(n\log n)$, which again implies that their combined depth is $O(n\log n)$.

  It remains to analyze the time spent on $\textsc{Split-Leaf}(J_i,u,v,\phi,\{\phi_1,\ldots,\phi_k\})$. Recall that $\textsc{Split-Leaf}$ finds the SCC-decomposition $T_{\phi'}$ of $J$ that contains $\phi$ as an external node and splits $\phi$ into $\phi_1,\ldots,\phi_k$ whereafter it fixes $T_{\phi'}$ if necessary. If $\phi$ has no parent then we simply replace $\phi$ by $\phi_1,\ldots,\phi_k$ as independent SCCs without changing the edges adjacent to $\phi_1,\ldots,\phi_k$. Otherwise we replace $\phi$ by $\phi_1,\ldots,\phi_k$ in $G^{\con}_{par(\phi)}$ and update the edges adjacent to $\phi_1,\ldots,\phi_k$. Note that before doing the split, $\phi$ represents a graph that contains all the graphs $G_{\phi_1},\ldots,G_{\phi_k}$ as subgraphs. In particular all connections to the rest of $G_{par(\phi)}$ are already present in the graph. The split causes $\phi$ to break, and the relevant edges to be transferred to the new vertices. For one of the vertices we can however reuse the connections from $\phi$. In particular our algorithm reuses the connections for the vertex $\phi_i \in \{\phi_1,\ldots,\phi_k\}$ whose associated graph $G_{\phi_i}$ contains the most vertices. The time it takes to split a vertex is therefore proportional to the number of edges adjacent to vertices of all the graphs $G_{\phi_1},\ldots,G_{\phi_{i-1}},G_{\phi_{i+1}},\ldots,G_{\phi_k}$.

  Note that all the new SCCs, excluding $G_{\phi_i}$, contain at most half as many vertices as $G_\phi$ originally did. Consider one of the SCC-decompositions $T$ that make up $J$. The external nodes of $T$ are a collection of disjoint subsets of $V(G)$, i.e., they contain at most $n$ vertices from the original graph. Since a split moves vertices to new nodes of half the size, each vertex $v$ can only be moved $O(\log n)$ times in $T$ by $\textsc{Split-Leaf}$. Each move takes time proportional to the number of edges adjacent to $v$, so the total time spent splitting leaves of $T$ is at most $O(m \log n)$. Since, by Lemma~\ref{lemma:num-scc-decompositions}, there are only $O(n)$ SCC-decompositions in $J$, it follows that the total time spent splitting leaves is $O(mn\log n)$.

  It remains to consider the time that $\textsc{Split-Leaf}$ spends on fixing the SCC-decomposition in the call $\textsc{Fix-SCC-decomposition}(T_{\phi'},u,v,\phi,\{\phi_1,\dots,\phi_{i-1},\phi_{i+1},\dots,\phi_k\})$. Here the work can however be charged to the depth reduction of the vertices that are moved, as it was the case in $\textsc{Delete-Edge-from-SCC-decomposition}(T_\phi,u,v)$. This completes the proof. (Note that charging the work to depth reduction was not possible when splitting a leaf since the depth did not change for the vertices that were moved.)
\end{proof}

We have skipped until now the details of how the matrix $A$ for answering queries is updated. As in the initialization of the joint SCC-decomposition, this is done when updating the topmost SCC-decompositions that each only contain a single internal node. Let $w$ be the internal node of such an SCC-decomposition $T_w$. When $\textsc{Fix-SCC-decomposition}$ or $\textsc{Split-Leaf}$ adds a new child $\phi$ to $w$ in $T_w$, $A[u,w]$ is updated for all the vertices $u \in V(G_\phi)$ to be equal to the index of this new child. $A[u,w]$ is similarly updated when $G_w$ breaks and a new SCC is formed, i.e., this SCC is independent of whether $w$ is deleted from $G$ or not. Note that updating $A$ does not affect the running time of our algorithm since this work is dominated by the work performed by $\textsc{Fix-SCC-decomposition}$ and $\textsc{Split-Leaf}$, i.e., we anyway run through all vertices of new SCCs of $G\setminus \{w\}$ to update edges to the rest of the condensed graph $G^{\con}_w$.

As described briefly in Section~\ref{sec:delete}, Łącki's SCC-decomposition can be implemented such that is uses $O(m+n)$ space~\cite{scc-decomposition}. Since a balanced joint SCC-decomposition consists of $O(n)$ SCC-decompositions (Lemma~\ref{lemma:num-scc-decompositions}), it follows that a naive implementation of our data structure uses $O(mn)$ space. In the next section we show how to obtain an alternative bound of $O(n^2\log n)$.

\subsection{Space-Efficient Representation of SCC-Decompositions for Dense Graphs}\label{sec:space}

Recall that when we delete an edge $(u, v)$, then we search for vertices with no out-going or in-coming edges in the condensed graph $G^{\con}_\phi$, where $\phi = \LCA(u,v)$ is the lowest common ancestor of $u$ and $v$. The key to improving the space bound is to observe that we do not need to explicitly store the edges to retrieve this information; it is enough to store the in- and out-degrees of the vertices of $G^{\con}_\phi$. The same is true when nodes are moved up in the tree of an SCC-decomposition by $\textsc{Fix-SCC-decomposition}$. For a vertex $v \in V(G^{\con}_\phi)$, we therefore let $\textsc{in-deg}(v)$ and $\textsc{out-deg}(v)$ be the in- and out-degree of $v$, respectively.

When we search for vertices that should be removed from a condensed graph, we repeatedly visit neighbors of vertices whose in- or out-degree has been reduced to zero. We here make use of the edges in $G^{\con}_\phi$, but Łącki's analysis actually allows us to spend time linear in the number of edges adjacent to the corresponding vertices in $G$. We use this observation to remove the edges from $G^{\con}_\phi$, and instead operate with the original edges in $G$. For every node $\phi$ of every SCC-decomposition in the joint SCC-decomposition, we therefore maintain a list $vertices(\phi)$ of all the vertices in the subgraph $G_\phi$. We also maintain a double-array $contains\langle \phi, v\rangle$ such that $contains\langle \phi, v\rangle$ is the vertex of $G^{\con}_\phi$ that contains $v$ or such that $contains\langle \phi, v\rangle = \textsc{Null}$ if $v$ does not appear in $G_\phi$. Since the total number of nodes of all SCC-decompositions is $O(n\log n)$ (Lemma~\ref{lemma:count}), the space needed to store $vertices$ and $contains$ is $O(n^2\log n)$. The information in $vertices$ and $contains$ is updated when nodes are moved in the SCC-decompositions. Note that we only update a single condensed graph at a time, and since the analysis allows us to spend time linear in the number of edges adjacent to vertices that are moved, it is straightforward to update $vertices$ and $contains$ within the same time bound.

Suppose we wish to delete an edge $(u, v)$ with $\phi = \LCA(u,v)$, and with $u'$ and $v'$ being the vertices of $G^{\con}_\phi$ that contain $u$ and $v$, respectively. We then reduce $\textsc{out-deg}(u')$ and $\textsc{in-deg}(v')$ by 1 each, and check whether either of them is reduced to zero. If, e.g., $\textsc{in-deg}(v')$ is reduced to zero then we scan through all vertices in $vertices(v')$ and collect all edges in $G$ that leave $v$. For each such edge $vw$, we retrieve $w' = contains\langle \phi, w\rangle$ and reduce $\textsc{in-deg}(w')$ by 1. The process is repeated if $\textsc{in-deg}(w')$ becomes zero. We can thus use the information in $vertices$ and $contains$ instead of the edges from $G^{\con}_\phi$. The information about edges in the condensed graphs dominates the space used to represent a joint SCC-decomposition, and by replacing the edges with the lists $vertices$ and the double-array $contains$, we reduce the space to $O(n^2\log n)$.

Recall that Łącki~\cite{scc-decomposition} finds the lowest common ancestor of the end-points of an edge $(u, v)$ by storing a pointer from the edge to $\LCA(u,v)$. This takes $O(m)$ space per SCC-decomposition, so we need another way to find $\LCA(u,v)$ that uses less space. Since a joint SCC-decomposition only contains $O(n\log n)$ nodes in total (Lemma~\ref{lemma:count}), we actually have time to visit all the nodes when looking for lowest common ancestors after being asked to delete an edge $(u, v)$. We can therefore locate the lowest nodes that contain $u$ and $v$, and move from there toward the root until finding their lowest common ancestor. The process can be simplified by storing for each vertex $v$ and each SCC-decomposition $T$ a pointer to the lowest node $\phi$ of $T$ whose subgraph $G_\phi$ contains $v$. Since there are $O(n)$ SCC-decompositions this takes $O(n^2)$ space.

Combining the above observations with Theorem~\ref{thm:mn} gives us the following theorem.

\begin{theorem}
	\label{lem:mn-space-efficient}
A balanced joint SCC-decomposition can be maintained under edge deletions using $O(mn\log n)$ total update time and using $O(n^2\log n)$ space.
\end{theorem}

\section{Strong connectivity under vertex failures} 

In this section we use the decremental joint SCC-decomposition in order device decremental algorithms for various connectivity notions defined with respect to vertex failures.

\subsection{Decrementally reporting SCCs in $G\setminus \{v\}$}
\label{sec:maintain-SCC-IDs}

In Section \ref{sec:joint} we showed that together with a joint SCC decomposition, in a total of $O(mn \log n)$ time and $O(n^2 \log n)$ space, we can also maintain an $n \times n$ matrix $A$ such that $A[u,v]$ is the index of the SCC of $G \setminus \{v\}$ that contains $u$.
This gives us an easy interface to test whether two vertices $u$ and $w$ are strongly connected in $G \setminus \{v\}$, by simply testing whether $A[u,v] = A[w,v]$.
We additionally maintain a data structure that can report each SCC in time proportional to its size, once its ID is specified. 
That is, for each vertex $v$, we maintain an array $L_v$ of doubly linked lists, such that $L_v[i]$ are the vertices in the SCC with ID $i$ in $G \setminus\{v\}$.
These lists can be easily updated as the matrix $A$ is updated, i.e., whenever an entry $A[i,j]$ changes from $\alpha$ to $\beta$, we remove $i$ from the $L_j[\alpha]$ and we add it to $L_j[\beta]$.
It straightforward to additionally monitor the size of those lists.
Using this data structure, we can report the vertices of an the SCC with index $i$ in $G\setminus \{v\}$ in time proportional to its size.

Notice that, whenever an SCC $C$ in $G \setminus \{v\}$ breaks into several SCCs $C_1, \dots, C_k$ we can report all the SCCs except of the SCC containing the most vertices in time proportional to their size.
Whenever $C$ breaks, the entries $A[u,v]$ for some vertices $u\in C$ change. 
The indexes of the SCC IDs are the initial value of $A[u,v]$, for any $u\in C$, and the new indexes that are assigned to the resulting SCCs.
Therefore, we can collect the resulting SCCs and reporting all of their vertices, except of the vertices contained in the resulting SCC with the most vertices.
We can further report all new resulting vertices that do not contain a specific vertex $w$ in time proportional to their size, by reporting all the reporting SCCs, except the one with index $A[w,v]$.

\subsection{Maintaining decrementally the dominator tree}
\label{sec:dominator-tree}

Let $G$ be a directed graph and let $s$ be the starting vertex from which we wish to maintain the dominator tree. 
We first produce a flow graph $G_s$ from $G$ by adding an edge from each vertex $v\in V\setminus \{s\}$ to $s$.
The addition of those edges has the following property. 
If a vertex $v$ is not strongly connected to $s$ in $G_s$, then there is no path from $s$ to $v$ in $G$.
Conversely, if a vertex $v$ is not strongly connected to $s$ in $G_s\setminus \{x\}$, while $s$ and $v$ are strongly connected in $G_s$, then all paths from $s$ to $v$ in $ G $ contain $x$.
That is, $x$ is a dominator of $v$ in $G$.

Dominance queries of the form ``does $x$ dominate $v$?'' can be answered by simply testing whether $s$ and $v$ are strongly connected in $G_s\setminus \{v\}$.
That is, we test if $A[v,x] = A[s,x]$, and if the answer is positive then $x$ dominates $v$. 
In this Section we show how to additionally maintain an explicit $O(n)$-size representation of all the dominance relations; that is, the dominator tree $D$ of $G_s$.
We assume that all vertices that become unreachable from $s$ are ignored for the rest of the algorithm.
Moreover, deletion of edges whose tail is unreachable from $s$ are ignored.
We assume that we have access to a data structure that reports all vertices that become unreachable from $s$ in $G_s$. 
This can be achieved by running independently a decremental algorithm for single-source reachability, i.e., in $O(mn)$ total update time and linear space \cite{ES81,HK95}.

Let $(x,y)$ be the edge to be deleted.
We denote by $D'$ the resulting dominator tree after the deletion of $(x,y)$.
We describe how to compute the updated dominator tree $D'$ after the deletion of $(x,y)$.
We denote by $depth(v)$ the depth of the $v$ in the dominator tree $D$; that is the number of edges in the path from $s$ to $v$ in $D$.
Respectively, we define as $depth'(v)$ to be the depth of $v$ in the dominator tree $D'$ (the dominator tree after the edge deletion).
As shown in Section \ref{sec:maintain-SCC-IDs}, the vertices that are no longer strongly connected to $s$ in $G_s\setminus \{v\}$, for any $v$, can be reported in time proportional to their number. 
These are exactly the vertices that are dominated by $v$ and were not dominated by $v$ before the deletion of $(x,y)$, as we explained above.
%
Let $N(v) = D'(v) \setminus D(v)$ be the set of vertices $w$ that are dominated by vertex $v$ in $G'_s$ but were not dominated by $v$ in $G_s$.
In this case, $v$ becomes an ancestor of $w$ in $D'$.
We can compute the depth in $D'$ for each vertex $w$ that acquired new dominators after the deletion by setting $\mathit{depth}'(w) \leftarrow \mathit{depth}(w) + | \{ v : w \in N(v)\}|$.
Next, we need to locate the parent $d'(w)$ in $D'$ of each vertex $w$. 
Notice that all vertices $\{ v : w \in N(v)\}$ become ancestors of $w$ in $D'$.
The parent $d'(w)$ in $D'$ of $w$ is the vertex $v$ with maximum $\mathit{depth}'(v)$ such that $w \in N(v)$.

To perform the above computations efficiently, we process each $N(v)$ set as computed by the corresponding SCC-decomposition data structure.
First, we increase the depth of each vertex $w \in N(v)$ by one. 
After processing all the sets $N(v)$ in this way we will have found the new depths of the vertices in $D'$.
Finally, we perform a second pass over the sets $N(v)$ in order to locate the unique parent $d'(w)$ for each vertex $w$.
To that end, for each $w$ we maintain a temporary variable $\widehat{d}(w)$, and initialize $\widehat{d}(w) \leftarrow d(w)$.
When we process a set $N(v)$ we update $\widehat{d}(w) \leftarrow v$ for all vertices $w \in N(v)$ such that $\mathit{depth}'(v) > \mathit{depth}'(\widehat{d}(w))$.
At the end of the second pass we will have $\widehat{d}(w) = d'(w)$ for all $w$, as desired.

Now we bound the total time required to maintain the dominator tree. 
Recall that our data structure can report each $N(v)$ set in $O(|N(v)|)$ time.
Hence, the running time of the above procedure, during the whole sequence of deletions, is bounded by the total size of all the sets $N(v)$.
Note that any vertex can appear in a specific $N(v)$ set at most once during the deletion sequence. 
Hence, the total size of all the sets $N(v)$ is $O(n^2)$.

\begin{lemma}
	The dominator tree of a directed graph $G$ with start vertex $s$ can be maintained decrementally in $O(mn \log n)$ total update time and $O(n^2 \log n)$ space, where $m$ is the number of edges in the initial graph and $n$ is the number of vertices.
\end{lemma}

%
%
%

\subsection{Answering decrementally strong connectivity queries under vertex failures}
\label{sec:queries-vertex-failures}

In this section we show how to answer various strong connectivity queries under single-vertex failures in optimal time while we maintain a directed graph decrementally.
More specifically, under any sequence of edge deletions, we consider answering the following types of queries:
\begin{itemize}
	\item[\it{(i)}] Report the total number of SCCs in $G\setminus \{v\}$, for a query vertex $v \in V$.
	\vspace{-.1cm}
	\item[\it{(ii)}] Report the size of the largest and of the smallest SCC in $G\setminus \{v\}$, for a query vertex $v\in V$.
	\item[\it{(iii)}] Report all the SCCs of $G\setminus \{v\}$, for a query vertex $v\in V$.
	\item[\it{(iv)}] Test if two query vertices $u$ and $w$ are strongly connected in $G\setminus \{v\}$, for a query vertex $v$.
	\item[\it{(v)}] For query vertices $u$ and $w$ that are strongly connected in $G$,
	report all vertices $v$ such that $u$ and $w$ are not strongly connected in $G\setminus \{v\}$ anymore.
\end{itemize}

For static strongly connected graphs, it was shown that after linear time preprocessing one can answer all of the above queries in optimal time \cite{GIP:SC}.
Here, we show how to preserve asymptotically optimal query time on a graph subject to edge deletions.
Before proving the main result of this section we need the following two technical lemmas. 
By definition of dominators, all paths from $s$ to vertices in $D(v)$, for any $v$, contain $v$.
Our first lemma shows that this property does not hold for the paths from $s$ but also for the paths from all vertices that are not in $D(v)$.

\begin{lemma}[\cite{2VCB}]
	\label{lemma:paths-through-SAP}
	Let $G$ be a digraph and $D$ be the dominator tree of the corresponding flow graphs $G_s$, for an arbitrary start vertex $s$.
	Suppose $v$ is a vertex such that $D(v)\setminus \{v\} \not= \emptyset$. 
	Then there is a path from a vertex $w\notin D(v)$ to $v$ in $G$ that contain no vertex in $D(v)\setminus \{v\}$. Moreover, all simple paths in $G$ from $w$ to any vertex in $D(v)$ contain $v$.
\end{lemma}

\begin{lemma}[\cite{GIP:SC}]
	\label{vertices-separators-ancestors}
	Let $u$ be a strong articulation point that is a separating vertex for vertices $x$ and $y$. Then $u$ must appear in at least one of the paths $D[s,x]$, $D[s,y]$, $D^{\R}[s,x]$, and $D^{\R}[s,y]$.
\end{lemma}

Now we are ready to prove the main result of this section.

\begin{lemma}
	We can maintain a digraph $G$ decrementally in $O(mn \log n)$ total update time and $O(n^2 \log n)$ space, where $m$ is the number of edges in the initial graph and $n$ is the number of vertices, so that after each edge deletion we can answer in asymptotically optimal time the following types of queries:
	\begin{itemize}
		\item[\it{(i)}] Report in $O(1)$ time the total number of SCCs in $G\setminus \{v\}$, for a query vertex $v \in V$.
		\vspace{-.1cm}
		\item[\it{(ii)}] Report in $O(1)$ time the size of the largest and of the smallest SCC in $G\setminus \{v\}$, for a query vertex $v\in V$.
		\item[\it{(iii)}] Report in $O(n)$ worst-case time all the SCCs of $G\setminus \{v\}$, for a query vertex $v\in V$.
		\item[\it{(iv)}] Test in $O(1)$ time if two query vertices $u$ and $w$ are strongly connected in $G\setminus \{v\}$, for a query vertex $v$.
		\item[\it{(v)}] For query vertices $u$ and $w$ that are strongly connected in $G$,
		report all vertices $v$ such that $u$ and $w$ are not strongly connected in $G\setminus \{v\}$, in asymptotically optimal worst-case time, i.e., in time $O(k)$, where $k$ is the number of separating vertices.
		(For $k=0$, the time is $O(1)$).
	\end{itemize}
\end{lemma}
\begin{proof}
	The queries \textit{(i),(iii),(iv)} can be answered by maintaining the matrix $A$, as shown in Section \ref{sec:maintain-SCC-IDs}.
	As we already mentioned in Section \ref{sec:maintain-SCC-IDs}, we can maintain for each $G\setminus \{v\}$ a list of SCCs.
	We also shown how we maintain the size of each SCC in $G\setminus \{v\}$.
	In order to have fast access to the minimum and the maximum size of the SCCs, we store the sizes of the SCCs in a min-heap and a max-heap.
	Those heaps can be updated in total time $O(n \log n)$ for each subgraph $G\setminus \{v\}$, for each $v$, as follows.
	Whenever an SCC breaks, we add the IDs of the newly created SCCs together with their size into the heaps, and we also update the size of the SCC that maintained the same ID.
	Since each time an SCC breaks at least two vertices are no longer in the same SCC, this can happen at most $n-1$ times.
	Moreover, there can be at most $n$ SCCs in a graph.
	Therefore, at most $O(n)$ insertions and updates are executed to each heap.
	This implies that the total time spend on maintaining each heap for each $G\setminus \{v\}$ is $O(n \log n)$, which sums up to $O(n^2\log n)$ for all $v$.
	Given one min-heap and one-max heap, we can answer type \textit{(ii)} queries in constant time.
	
	We are left to show how to answer queries of type \textit{(v)}.
	For this type of queries we will assume that we maintain the dominator tree $D$ of the graph and the dominator tree $D^{\R}$ of the reverse graph.
	By Lemma \ref{vertices-separators-ancestors}, all separating vertices for $u$ and $w$ are either ancestors of $u$ or $w$ in $D$ or ancestors of $u$ or $w$ in $D^{\R}$.
	We only show how to report the separating vertices for $u$ and $w$ that are ancestors of $u$ or $w$ in $D$ since the procedure for $D^{\R}$ is completely analogous.
	By Lemma \ref{lemma:paths-through-SAP}, notice that if there exist a vertex $z$ such that $u\in D(z)$ and $w\notin D(z)$ or such that $w\in D(z)$ and $u\notin D(z)$ then $z$ is a separating vertex for $u$ and $w$.
	That means, all vertices on the path from $\NCA_D(u,w)$ (the nearest common ancestor of $ u $ and $ w $ in $ D $) to $d(u)$ and from $\NCA_D(u,w)$ to $d(w)$ are separating vertices for $u$ and $w$.
	We can find and report all those vertices in asymptotically optimal time by simply following the parents of $u$ and $w$ in $D$ until they meet at $\NCA_D(u,w)$ (here we assume that also the depth of each vertex in $D$ is available, otherwise we can compute the depth for all vertices in $O(n)$ time after each edge deletion).
	
	Next, we show that all vertices $z \notin D(\NCA_D(u,w))$ that separate $u$ and $w$ appear on a path of $D$.
	More specifically, there is a vertex $z\in V\cup \{\emptyset\} \setminus\{D(\NCA_D(u,w))\}$, such that all vertices on the path from $z$ to $\NCA_D(u,w)$ in $D$ are separating vertices for $u$ and $w$.
	Let $z$ be the first vertex on the path from $s$ to $\NCA_D(u,w)$ on $D$ that separates $u$ and $w$.
	If there is no such vertex $z$, then none of the vertices on the path from $s$ to $\NCA_D(u,w)$ is a separating vertex for $u$ and $w$ and we are done.
	We can verify the existence of such $z$ by testing if $\NCA_D(u,w)$ is a separating vertex for $u$ and $w$ (this can be done in constant time by executing one type \textit{(iv)} query).
	Assume now that $z\not= \emptyset$.
	By Lemma \ref{lemma:paths-through-SAP}, either all paths from $u$ to $w$ contain $z$ or all paths from $w$ to $u$ contain $z$.
	Assume, without loss of generality, that all paths from $u$ to $w$ contain $z$.
	Since $z\notin D(\NCA_D(u,w))$, all paths from $z$ to $w$ contain all vertices on the path from $z$ to $w$ in $D$ (including all vertices on the path from $z$ to $\NCA_D(u,w)$).
	This allows us to efficiently identify and report all separating vertices $z\notin D(\NCA_D(u,w))$ for $u$ and $w$ as follows.
	If $z\not=\emptyset$, we start testing the vertices on the path from $\NCA_D(u,w)$ to $s$ in $D$, reporting all vertices that are separating vertices for $u$ and $w$, and once we find a vertex that is not a separating vertex for $u$ and $w$ we stop (as we proved, there are not further vertices on the path from $s$ to $\NCA_D(u,w)$).
	
	Notice that we only spend time proportional to the vertices that we report, the computation of $\NCA_D(u,w)$, and only a single type \textit{(iv)} query that does not report a vertex.
	We also spend the same time on the dominator tree $D^{\R}$ of the reverse graph.
	We only need to be careful not to report the same vertex twice, which can be trivially implemented within the claimed time bounds.
\end{proof}

\subsection{Maintaining decrementally the vertex-resilient components}
\label{sec:VRBs}
In this section we show how we can maintain the vertex-resilient components of a directed graph.
By definition, two vertices $u$ and $w$ are vertex-resilient if and only if there is no vertex $v$ such that $u$ and $w$ are not strongly connected in $G\setminus \{v\}$.
In our algorithm we will be testing this property after every edge deletion, and whenever we identify a vertex-resilient component $B$ containing vertices from different SCCs in $G\setminus \{v\}$, for some $v$, then we refine the vertex-resilient component according to these components.
Assume that a vertex-resilient component $B$ breaks after an edge deletion.
That is, there is a vertex $v$ such that vertices of $B$ lie in different SCCs in $G\setminus \{v\}$.
Let $C_1,C_2,\dots, C_k$ be the SCCs in $G\setminus \{v\}$, then we replace $B$ by $\{B\cap C_1\} \cup \{v\} , \{B\cap C_1\} \cup \{v\}, \dots, \{B \cap C_k\} \cup \{v\}$. 
These refinements can be easily carried out in $O(n)$ time, and therefore we spend total time $O(n^2)$ for this part.

Now we show how to efficiently detect whether two vertex-resilient vertices appear in different SCCs in $G\setminus \{v\}$, for some $v$.
Whenever an SCC $C$ in $G\setminus \{v\}$, for some $v$, breaks into many SCCs $C_1,\dots, C_k$ the vertices of all SCCs except of one can be listed in time proportional to their edges as shown in Section \ref{sec:maintain-SCC-IDs}.
Without loss of generality, let $C_1, \dots, C_{k-1}$ be those SCCs.
For SCC $C_i$, for $1\leq i \leq k-1$, we examine whether the vertex-resilient components containing subsets of vertices of $C_i$ are entirely contained in $C_i$.
This can be easily done in time proportional to $|C_i|$.
Notice that we do not examine the vertices in $C_k$.
We claims that, if we do not find a vertex-resilient pair that is disconnected in $G\setminus \{v\}$ by the searches in $C_i, 1\leq i\leq k-1$, then there is no such pair.
Indeed, assume that there is a pair of vertex-resilient vertices $x,y$ such that $x\in C_k, y\notin C_k$. 
By the fact that $x$ and $y$ were vertex-resilient before the edge deletion it follows that $y\in C$, and therefore $y\in C_i, 1\leq i\leq k-1$, in which case we would find this by searching in $C_i$.
If we detect a vertex-resilient component whose vertices lie in different SCCs in $G\setminus \{v\}$, for some $v$, then we perform the refinement phase in $O(n)$ time.

Notice that all the tests that we described above (excluding the time for the refinement operations) can be executed in time proportional to the number of vertices of a broken SCC in $G\setminus \{v\}$, for some $v$, and that are not contained in the largest resulting SCC.
Each vertex can appear at most $\log n$ times in an SCC of $G\setminus \{v\}$ that is not the largest after a big SCC breaks. 
That means we spend $O(n \log n)$ time for each graph $G\setminus \{v\}$ on the aforementioned queries, for some vertex $v$, and thus, $O(n^2 \log n)$ in total.
Thus, we have the following lemma.

\begin{lemma}
	\label{lemma:VRB-running-time}
	The vertex-resilient components of a directed graph $G$  can be maintained decrementally in $O(mn \log n)$ total update time and $O(n^2 \log n)$ space, where $m$ is the number of edges in the initial graph and $n$ is the number of vertices.
\end{lemma}

\subsection{Maintaining decrementally the maximal $2$-vertex-connected subgraphs}

In this section we show how to maintain the maximal $2$-vertex-connected subgraphs of a directed graph decrementally.
Recall from the Introduction, that a graph is $2$-vertex-connected if it has at least three vertices and for each $v\in V$ it holds that $G\setminus \{v\}$ remains strongly connected.
We will allow for degenerate maximal $2$-vertex-connected subgraphs consisting of two mutually adjacent vertices.
Maximal $2$-vertex-connected subgraphs do not induce a partition of the vertices of the graph.
More specifically, two maximal $2$-vertex-connected subgraphs might share at most one common vertex.
First, we describe a simple-minded algorithm that computes the maximal $2$-vertex-connected subgraphs of a directed graph and later we show how to dynamize this algorithm.
The simple-minded algorithm proceeds as follows.
Compute for each vertex $v\in V$ the SCCs of $G\setminus \{v\}$ and remove all edges between the different SCCs that are created.
The algorithm is repeated until no more edges are removed from the graph after examining all vertices.
Let $G'$ be the resulting graph obtained at the end of this process..
The vertex-resilient components of $G'$ of size more than $2$ are the non-degenerate maximal $2$-vertex-connected subgraphs of $G$, as the following lemma shows.

\begin{lemma}
	\label{lem:VRBsAND2VCSs}
	Graph $G'$ has the same $2$-vertex-connected subgraphs as $G$ (including the degenerate $2$-vertex-connected subgraphs), and there is no edge between vertices that are not in the same  $2$-vertex-connected subgraph.
	Moreover, the vertex-resilient components of $G'$ are equal to its maximal $2$-vertex-connected subgraphs.
\end{lemma}
\begin{proof}
	We begin with the first part of the lemma, that is, we show that $G'$ has the same $2$-vertex-connected subgraphs as $G$.
	First, notice that $G'$ cannot contain a $2$-vertex-connected subgraph that is not a $2$-vertex-connected subgraph in $G$ since $G'$ is a subgraph of $G$.
	Now we show that each $2$-vertex-connected subgraph of $G$ is a $2$-vertex-connected subgraph of $G'$.
	To do so, we show that no edge $(x,y)$ is removed such that $x$ and $y$ are in the same $2$-vertex-connected subgraph $C$.
	Assume by contradiction, that this is not true, and $(x,y)$ is the first edge that is removed from~$C$.
	This implies that, throughout the procedure of the algorithm that removes edges there is an instance $\widetilde{G}$ and a vertex $w$ such that $x$ and $y$ are in different SCCs in $\widetilde{G}\setminus \{w\}$.
	This cannot happen if $w=x$ or $w=y$, since in those cases $x$ and $y$ do not appear in $\widetilde{G}\setminus \{x\}$ or $\widetilde{G} \setminus \{y\}$, respectively.
	If $w \notin C$, then this should not happen as $C$ is strongly connected and we assume that $(x,y)$ is the first edge deleted from $C$.
	If on the other hand $w \in C$, then $x$ and $y$ should be in the same SCC in $\widetilde{G}\setminus \{w\}$ by the definition of the $2$-vertex-connected subgraph $C$, clearly a contradiction.
	
	Now we prove the second part of the lemma: namely, there is no edge between vertices that are not in the same $2$-vertex-connected subgraphs.
	Assume by contradiction that such an edge $(x,y)$ exists.
	In particular, we show that for any edge $(x,y)$, $x$ and $y$ are in the same $2$-vertex-connected component.
	Since there is no vertex $w$ in the graph disconnecting $x$ and $y$, $x$ and $y$ are vertex-resilient.
	Let $C$ be the set of all the vertices that are vertex-resilient to both $x$ and $y$ in $G'$.
	Now we show that for each vertex $v\in C$ all vertices $C\setminus \{v\}$ remain strongly connected in $G'[C\setminus \{v\}]$.
	Assume not: then there is a pair of vertices $u,w \in C$ that are strongly connected in $G'\setminus \{v\}$ (which follows from the fact that they are vertex-resilient) but not in $G'[C\setminus \{v\}]$.
	Then either all paths from $u$ to $w$ or all paths from $w$ to $u$ in $G'\setminus \{v\}$ contain vertices in $V\setminus C$.
	Without loss of generality, let $P$ be any such path from $u$ to $w$, and let $z\in P \cap \{V\setminus C\}$.
	Then, by the definition of vertex-resilient components there exists a vertex $q$ such that $u$ (as well as $w$) and $z$ are not strongly connected in $G'\setminus \{q\}$.
	As a result, the procedure that removes edges and generated $G'$ would eliminate all paths from $z$ to $u$ and all paths from $w$ to~$z$ that do not contain $q$ (since $x$ and $y$ are not in the same SCC with $z$ in $G'\setminus \{q\}$). 
	This implies that $P$ is not a simple path as it contains $u$, then $q$, then $z$, then again $q$, and finally $w$.
	Therefore, there is a path from $x$ to $q$ and a path from $q$ to $y$ avoiding $z$ (and also $v$ since $P$ is a path in $G'\setminus \{v\}$).
	If $q\in \{V\setminus C\}$, then we repeat the same argument for $z=q$.
	If we apply the same argument for every path from $u$ to $w$ and for each vertex $z \notin C$, it follows that there exists a path from $u$ to $w$ in $G'\setminus \{v\}$ avoiding all vertices in $V\setminus C$.
	This contradicts the assumption that all paths from $u$ to $w$ in $G'\setminus \{v\}$ contain vertices in $V\setminus C$.
	Therefore,  for each vertex $v\in C$ all vertices $C\setminus \{v\}$ remain strongly connected in $G'[C\setminus \{v\}]$.
	Then, $C$ satisfies the definition of a $2$-vertex-connected subgraph, which contradicts the fact that $x$ and $y$ are not in the same $2$-vertex-connected subgraph.
	Thus, all edges are between vertex-resilient components of $G'$.
	
	Notice that we showed the above for any edge $(x,y)$.
	Hence, we also proved that every vertex-resilient component is fully contained inside a $2$-vertex-connected subgraph.
	By definition, each (maximal) $2$-vertex-connected subgraph of $G'$ is fully contained in a vertex-resilient component of $G'$.
	Thus, the vertex-resilient components in $G'$ are equal to the maximal $2$-vertex-connected subgraphs of $G'$.
\end{proof}

We now present our decremental algorithm for maintaining the maximal $2$-vertex-connected subgraphs of $G$.
Our algorithms is a simple extension of the simple-minded static algorithm for computing the maximal $2$-vertex-connected subgraphs.
More specifically, we will maintain $G'$ decrementally, and we will additionally be running the decremental vertex-resilient components algorithm on $G'$ simultaneously.
In order to do so, we might be deleting from the maintained graph more edges than dictated by the sequence of edge deletions on $G$.

First, we initialize a joint SCC-decomposition, as described in Section \ref{sec:joint}.
We will maintain a subgraph $G'$ of the current version of $G$; by current version of $G$ we mean the initial graph minus the edges that were deleted from it.
Let $v$ be the root of an SCC-decomposition.
Initially, we collect all edges among different SCCs in $G'\setminus \{v\}$ and remove them from $G' $ by executing them as additional edge deletions.
We assume that these edges are marked and are inserted into a global set data structure $L$ and after the necessary updates in the joint SCC-decomposition they are deleted from $G'$ one at a time, making sure every edge appears at most once in $L$.
Whenever an SCC in $G'\setminus \{v\}$ breaks into several SCCs $C_1,\dots, C_k$ after an edge deletion, we collect all the unmarked edges between different SCCs and add them into the global set data structure $L$ so they will be deleted from the graph.
If some vertices in $G'$ are no longer strongly connected to $v$, then we simply ignore them from the SCC-decomposition rooted at $v$. (Notice that $v$ can no longer affect the strong connectivity between vertices that are not strongly connected to $v$.)
We do the same for each $v \in V$.
The above procedure maintains $G'$ under any sequence of edge deletions, since for each vertex we only remove edges between different SCCs in $G' \setminus \{v\}$ and when the auxiliary edge deletions end there are no edges between different SCCs in $G' \setminus \{v\}$.

In order to maintain decrementally the maximal $2$-vertex-connected subgraphs, we additionally run in the algorithm from Section \ref{sec:VRBs} for maintaining decrementally the vertex-resilient components of size more than $2$ on the side. 
To implement the size restriction, we simply ignore all vertex-resilient components of size at most $2$.
The correctness of the algorithm follows from Lemma \ref{lem:VRBsAND2VCSs}.

Now we bound the running time of the algorithm.
The time for handling all edge deletions in the maintained graph $G'$ is bounded by $O(mn \log n)$ by Lemma \ref{lem:mn-space-efficient}.
Moreover, the total time spent to maintain the vertex-resilient components of $G'$ is $O(mn \log n)$ by Lemma \ref{lemma:VRB-running-time}.
We only need to bound the time we spend to collect all edges among different resulting SCCs after some SCC in $G' \setminus \{v\}$ breaks, for any $v$.
We do that as follows.
Whenever an SCC $C$ breaks into several SCCs $C_1,\dots, C_k$ in $G' \setminus \{v\}$, for some $v$, we only need to identify all the edges among $C_1, \dots, C_k$, since the algorithm previously should have removed all edges from $C$ to other SCCs in $G'\setminus \{v\}$.
Without loss of generality, let $C_k$ be the largest SCC among $C_1,\dots, C_k$.
For each $C_i$, $1\leq i \leq k-1$ we iterate over all the unmarked edges incident to their vertices and test whether their endpoints are in different SCCs in $G'\setminus v$, and if yes we mark them and insert them into the global set data structure $L$ in order to delete them from the graph.
Notice that whenever the algorithm iterates over the edges of a vertex, that vertex is contained in an SCC in $G' \setminus v$ that is at most half the size of its previous SCC.
That means each vertex will be listed at most $\log n$ times.
Therefore, for each $v$, we consider the edges of each vertex at most $\log n$ times, and therefore at most $n\log n$ times in total, for all $v$.
Hence we spend at most $O(mn\log n)$ to collect all edges among different resulting SCCs after some SCC in $G' \setminus \{v\}$ breaks, for any $v$.
Thus, we have the following lemma.

\begin{lemma}
	\label{lem:2VCCs-running-time}
	The maximal $2$-vertex-connected subgraphs of a directed graph $G$  can be maintained decrementally in $O(mn \log n)$ total update time and $O(n^2 \log n)$ space, where $m$ is the number of edges in the initial graph and $n$ is the number of vertices.
\end{lemma}

\section{Strong connectivity under edge failures}

In this Section we consider several application our joint SCC-decomposition that are related to strong connectivity under single edge failures, and to $2$-edge-connectivity.
In order to devise efficient algorithm for the problems that we consider, we need first to be able to maintain all the strong bridges of a graph (the edges whose deletion affects the strong connectivity of the graph).
To achieve the latter, we in turn first show how to maintain both an in- and out-dominator tree at some arbitrary vertex $ s $, where in our case we choose $ s $ uniformly at random for reasons of efficiency.

\subsection{Maintaining decrementally an in-out dominator tree in each SCC} 
\label{sec:in-out-dominators}
Let $D$ be a dominator tree of a directed graph $G$ and $D^{\R}$ be a dominator tree of the reverse graph $G^{\R}$, both rooted at the same starting vertex $s$.
We call such a pair of dominator trees as an in-out dominator tree, rooted at $s$.
In this section we will show how to maintain an in-out dominator tree in each SCC of $G$ decrementally.
More specifically, we present a randomized algorithm that runs in $O(mn \log n)$ expected total update time, and maintain an in-out dominator tree in each SCC, rooted at a vertex chosen uniformly at random.

Recall that the algorithm from Section \ref{sec:dominator-tree} maintains decrementally a dominator tree rooted at an arbitrary root in $O(mn \log n)$ total time and $O(n^2 \log n)$ space.
Given a graph $G$, and a starting vertex $s$, we decrementally maintain a in-out dominator tree from $s$ as follows.
We maintain a dominator tree $D$ of $G$ and a dominator tree $D^{\R}$ of $G^{\R}$, both rooted at $s$, using the algorithm from Section \ref{sec:dominator-tree}.
That is, we apply each edge deletion to both instances of the decremental dominators algorithm.
The two algorithms might run on different sets of vertices throughout their execution, since the set of reachable vertices from $s$ might differ from the set of vertices that reach $s$.
This required $O(mn \log n)$ time and $O(n^2 \log n)$ space in total, by the fact that we simply run two instances of the decremental dominators algorithm simultaneously.
Bellow, we refer to the maintainance of an in-out dominator tree of a graph without getting into the details of the two instances that we maintain.

Our algorithm works as follows. 
For each SCC $C$ of $G$, we randomly choose a starting vertex $s\in C$, and maintain an in-out dominator tree rooted at $s$.
Whenever an SCC $C$, with starting vertex $s$, breaks into several SCCs $C_1,\dots, C_l$ we proceed as follows.
Let, without loss of generality, $C_l$ be the new resulting SCC that contains $s$.
We update the in-out dominator tree, rooted at $s$, by removing from $G[C]$ all edges $(u,v)\in C_i \times C_j, i\not= j$.
The resulting in-out dominator tree is a valid for $G[C_l]$.
Furthermore, for each SCC $C_i, 1\leq i \leq l-1$, we randomly choose a starting vertex $s_i$ and we maintain decrementally an in-out dominator tree of $G[C_i]$, rooted at $s_i$.

The correctness of the aforementioned algorithm follows from the correction of the decremental dominators algorithm from Section \ref{sec:dominator-tree}.
We next analyze the running time of the algorithm in the following lemma.
We use the type of analysis that was used in \cite{RodittyZ08}, where the authors presented an algorithm with $O(mn)$ total expected running time for the problem of decrementally maintaining the SCCs of a directed graph.

\begin{lemma}
Let $G$ be a directed graph.
An in-out dominator tree in each SCC of $G$, rooted at a vertex chosen uniformly at random, can be maintained in $O(mn \log n)$ total expected time against an oblivious adversary, using $O(n^2 \log n)$ space.
\end{lemma}
\begin{proof}
By Lemma \ref{lem:mn-space-efficient}, the algorithm for maintaining the in-out dominator tree of a directed graph $G$ with respect to an arbitrary root, runs in total $O(mn \log n)$ time and $O(n^2 \log n)$ space.
In order to simplify our analysis we will charge all its running time at the beginning of the algorithm.
The algorithm can be easily extended to maintain the in-out dominator tree of the SCC $C$ containing $s$, as follows.
We initiate the algorithm to run on $G[C]$
Whenever $C$ breaks after an edge deletion, we remove from the graph all edges that contain at least one endpoint that is not contained in the new SCC of $s$.
Notice that this additional edge deletion does not affect the total running time of the algorithm as we already charged the algorithm with any possible sequence of edge deletions.
The space total requirement is $O(n^2 \log n)$ since the total number of vertices in all SCCs remains $n$, independently of the number of the SCCs.
In order to simplify the analysis, we assume that the running time of the algorithm for maintaining an in-out dominator tree is exactly $mn \log n$.

Now let $f(m,n)$ be the expected running time of the algorithm for maintaining the in-out dominator tree in a strongly connected graph (if the graph is not strongly connected then we assume we refer to any SCC of the graph), rooted at a vertex $s$ chosen uniformly at random.
We have
\[
f(m,n) = mn \log n + \sum_{i=1}^{l} f(m_i,n_i) - \sum_{i=1}^{l} \frac{n_i}{n} m_in_i \log n_i
\]
for some $l\geq 2$, where $l$ is the number of SCCs when the graph is no longer strongly connected, and $m_1, \dots, m_l$ and $n_1, \dots, n_l$ refer to the number of edges and vertices of those SCCs, respectively.
Obviously, $n_i \geq 1, m_i \geq 0, 1\leq i \leq l$, and moreover $\sum_{i=1}^{l} m_i < m$ (since at least one edge should had endpoints from two different SCCs in the initial graph) and $\sum_{i=1}^{l} n_i = n$.
The term $mn \log n$ represents the total time spent by the algorithm for maintaining the in-out dominator tree rooted at $s$ in its SCC.
The term $\sum_{i=1}^{l} f(m_i,n_i)$ represents the total time spent by the decremental dominators algorithm in each of the $l$ SCCs that are created after the graph is no longer strongly connected.
However, the first term already covers for the work done in all future SCCs containing $s$ as we can keep using the same instance of the algorithm.
Therefore, the time spent in the new SCC of $s$ after the graph breaks into several SCCs is already covered and we have to subtract it, since we accounted for it in the second term.
Notice that $s$ can be in the $i$-th SCC with probability $n_i/n$ as it was chosen uniformly at random, and therefore with probability $n_i/n$ the algorithm will not spend $m_in_i \log n_i$ time in order to maintain the in-out dominator tree in $C_i$.
We point out that the fact that the graph breaks into the specific $l$ SCCs (or that it breaks in the first place) does not depend on any choice made by the algorithm, including the random choice of the starting vertex $s$.
Thus, the aforementioned formula correctly captures the expected running time of the algorithm.

We now show that $f(m,n) \leq 2 m n \log n$.
Our proof is by induction on $f(m,n)$.
The basis of the induction is the case where $n=1,m=0$, for which we have $f(m,n)=0$.
Assume now that the claim holds for every $(m',n')$ where $m'<m$ and $n'<n$.
Our goal is to show that 
\[
\sum_{i=1}^{l} 2 m_in_i \log n_i - \sum_{i=1}^{l} \frac{n_i}{n} m_in_i \log n_i \leq mn \log n.
\]
We set $x_i=m_i/m, y_i=n_i/n, z_i = \log n_i / \log n$, for which it holds that $0\leq x_i < 1,0 < y_i< 1, 0\leq z_i<1$.
Moreover, it holds that $\sum_{i=1}^{l}x_i < 1$, and by the fact that $z_i<1, \forall 1\leq i \leq l$ it also holds that $\sum_{i=1}^{l}x_iz_i < 1$.
We devide both sides of the inequality by $mn\log n$, and now we are left to show
\[
2\sum_{i=1}^{l} x_iy_iz_i - \sum_{i=1}^{l} x_iy^2_iz_i \leq \sum_{i=1}^{l}x_iz_i < 1.
\]
Moving everything on the right side we get
\[
\sum_{i=1}^{l}x_iz_i - 2\sum_{i=1}^{l} x_iy_iz_i + \sum_{i=1}^{l} x_iy^2_iz_i  = \sum_{i=1}^{l}x_iz_i(1-y_i)^2\geq 0
\] 
which obviously holds.
\end{proof}

\subsection{Maintaining decrementally the strong bridges of a digraph}

Let $G$ be a digraph.
We show how to maintain decrementally the strong bridges in each SCC of $G$. 
Note that edges whose endpoints are not both in the same SCC, cannot be strong bridges since their removal cannot affect the strong connectivity of $G$.

Our algorithm uses the decremental algorithm for maintaining an in-out dominator tree of a subgraph induced by an SCC, presented in Section \ref{sec:in-out-dominators}.
We only use this algorithm in order to have access to the dominator trees $D$ and $D^{\R}$ of each SCC (rooted at a randomly chosen vertex), and to the set of affected vertices (that change parent in in the dominator tree after a deletion) in each of the dominator trees after every edge deletion.
Note that the latter can be easily made available even without modifying the algorithm to report the affected vertices, as one can get this information by simply comparing the dominator tree before and after an edge deletion in $O(n)$ time.

The following lemma shows that all strong bridges of a strongly connected graph $G$ appear as edges of either the dominator tree $D$ of $G$ or the dominator tree $D^{\R}$ of $G^{\R}$, both rooted at the same arbitrary start vertex $s$.

\begin{lemma}[\cite{Italiano2012}]
	\label{lem:strong-bridges-on-dominator-trees}
	Let $G$ be a strongly connected graph, and let $D$ and $D^{\R}$ be the dominator tree of $G$ and $G^{\R}$, respectively, both rooted at the same arbitrary start vertex $s$.
	An edge $(u,v)$ is a strong bridge of $G$ only if $u=d(v)$ or $d^{\R}(u)=v$.
\end{lemma} 

Given the dominator trees $D$ and $D^{\R}$ of the strongly connected digraph $G$, rooted at an arbitrary start vertex $s$, we can compute the strong bridges of $G$ in time $O(m+n)$ by testing for each vertex the condition given by the following lemma.

\begin{lemma}[\cite{Italiano2012}]
\label{lem:strong-bridge-necessary-sufficient-condition}
Let $G$ be a strongly connected graph, and let $D$ and $D^{\R}$ be the dominator tree of $G$ and $G^{\R}$, respectively, both rooted at the same arbitrary start vertex $s$.
An edge $(d(v),v)$ is a strong bridge of $G$ if and only if for all $ w\in \{q:(q,v) \in E, q \not=d(v)\}$ it holds that $w\in D(v)$.
An edge $(d^{\R}(v),v)$ is a strong bridge of $G$ if and only if for all $ w\in \{q:(q,v) \in E^{\R}, q \not=d^{\R}(v)\}$ it holds that $w\in D^{\R}(v)$.
\end{lemma} 

Our algorithm for maintaining the strong bridges of each SCC of a directed graph uses a simple modification of the condition of Lemma \ref{lem:strong-bridge-necessary-sufficient-condition} which we are able to dynamize.
Let $(x,y)$ be the edge to be deleted from the graph,
We denote by $D'$ the dominator tree $D$ after the deletion of $(x,y)$.
In general we use the notation $f'$ to refer to a relation $f$ after the deletion of $(x,y)$.

Let $G$ be a strongly connected graph, or an SCC of a directed graph.
Moreover, let $D$ and $D^{\R}$ be the dominator trees of $G$ and $G^{\R}$, respectively, rooted at an arbitrary start vertex $s$.
We maintain for each vertex $w$ a counter $c(w)$ counting the number of incoming edges $(z,w)$ to $w$ in $G$ such that $z\notin D(w)$.
Analogously, we maintain for each vertex $w$ a counter $c^{\R}(w)$ counting the number of incoming edges $(z,w)$ to $w$ in $G^{\R}$ such that $z\notin D^{\R}(w)$.
Lemma \ref{lem:strong-bridge-necessary-sufficient-condition} suggests that an edge $(d(w),w)$ (resp., $(w,d^{\R}(w))$) is a strong bridge if and only if $c(w)=1$ (resp., $c^{\R}(w)=1$).
Our goal is to update those counters, as the graph undergoes edge deletions.
We only describe the process of updating the counters $c(w)$ for each vertex $w$, as the analogous process is used to update the counters $c^{\R}(w)$.

In order to simplify our algorithm we recompute from scratch the counters $c(w)$ for each vertex $w$ whenever the SCC of $s$ breaks.
Assume that $s$ was chosen as the root when $s$ was part of an SCC with $n'$ vertices and $m'$ edges.
The in-out dominator tree of the SCC of $s$, rooted at $s$, is maintained in $O(m'n' \log n')$ total update time.
Notice that we can compute the counters $c(w)$ for all $w$ in $O(m'+n')$ time, by simply traversing all the incoming edges of a vertex and applying Lemma \ref{lem:strong-bridge-necessary-sufficient-condition}.
Moreover the SCC of $s$ can break at most $O(n')$ times.
Therefore, these recomputations require total time $O(m'n')$, which we can charge to the algorithm for maintaining the in-out dominator tree rooted at $s$.
Thus, in what follows we can assume that after the deletion of an edge $(x,y)$ both $x$ and $y$ are strongly connected to $s$.

Now we show how to update the counters $c(w)$ for each $w$ after the deletion of an edge $(x,y)$.
First, we collect in $S$ all the affected vertices in $D$ (that is, the vertices that change parent in $D$) and their descendants in $D$.
For each vertex $z\in S$, and for each edge $(z,w)$ such that $w\notin S$ and $z \notin D(w)$, we set $c(w)=c(w)-1$.
Notice that the vertices in $V\setminus S$ retain their parent-descendant relations in $D'$.
Therefore, the counters $c(w)$ correctly contain the number of incoming edges $(z,w)$ to $w$, such that $z \notin S$ and $z\notin D'(w)$.
We are left to add to $c(w)$ the incoming edges $(z,w)$ to $w$, such that $z \in S$ and $z\notin D'(w)$.
We do that as follows.
For each vertex $z\in S$, and for each edge $(z,w)$ such that $w\notin S$ and $z\notin D'(w)$, we set $c(w)=c(w)+1$.
This completes the update of the counter $c(w)$ for each vertex $w\notin S$.
For the vertices $w\in S$, we simply iterate over the incoming edges to $w$ and execute the test of Lemma~\ref{lem:strong-bridge-necessary-sufficient-condition}.
Now an edge $(d(w),w)$ is a strong bridge if and only if $c(w)=1$. 

Notice that each time a vertex is either affected in $D$ or a descendant of an affected edge it increases its depth in the dominator tree.
Therefore, a vertex can be included in $S$ at most $O(n)$ times.
That means, we iterate over the edges of a vertex at most $O(n)$ times.
Hence, we spend $O(mn)$ time in total to update the counters $c(w)$ for all vertices $w$.
We analogously update the counters $c^{\R}(w)$ for each vertex $w$.
Thus, we have the following lemma.

\begin{lemma}
The strong bridges of a directed graph $G$, with $m$ edges and $n$ vertices, can be maintained decrementally in $O(mn \log n)$ total expected time against an oblivious adversary, using $O(n^2 \log n)$ space.
\end{lemma}

\subsection{Maintaining decrementally the SCCs in $G\setminus e$ for each $e$}
\label{sec:maintain-SCC-IDs-edge-failures}

We show how to maintain the SCCs in $G\setminus e$, for each strong bridge $e$, under any sequence of edge deletions.
Notice that if an edge $e$ is not a strong bridge then the SCCs in $G \setminus e$ are equal to the SCCs in $G$.
The following lemma shows the relation between the SCCs of $G\setminus e$, for some strong bridge $e=(u,v)$, and the SCCs of $G\setminus \{u\}$ and $G\setminus \{v\}$.
This will be helpful since we already shown (in Section \ref{sec:maintain-SCC-IDs}) how to maintain the SCCs of $G \setminus \{v\}$, for all vertices $v$.

\begin{lemma}
\label{lem:SCC-partition}
Let $G$ be a strongly connected graph and let $e=(u,v)$ be a strong bridge of $G$. Two vertices are strongly connected in $G \setminus e$ if and only if they are strongly connected in either $G\setminus \{u\}$ or $G \setminus \{v\}$.
Moreover, let $C_u, C_v$ be the SCCs of $G\setminus e$ containing $u$ and $v$, respectively.
All SCCs of $G\setminus e$, except $C_u$ and $C_v$, are SCCs of both $G\setminus \{u\}$ and $G \setminus \{v\}$.
\end{lemma}
\begin{proof}
If two vertices $w$ and $z$ are not strongly connected in $G\setminus e$, then all paths from $w$ to $z$ contain $e$, and therefore both its endpoints $u$ and $v$.
Therefore, $w$ and $z$ are not strongly connected in neither $G \setminus \{u\}$ nor $G\setminus \{v\}$.
First, note that if in $G\setminus e$, $u$ and $v$ are not strongly connected. 
Now assume that $w$ and $z$ are strongly connected in $G \setminus e$.
Let $C$ be the SCC $C$ of $G \setminus e$ containing $w$ and $z$.
Clearly, $C$ cannot contain both $u$ and $v$ since $(w,z)$ is a strong bridge.
If $C$ does not contain any of $u$ or $v$, then $C$ is an SCC of both $G\setminus \{u\}$ and $G \setminus \{v\}$, since all the other vertices of $C$ and the edges among them remain in both $G \setminus \{u\}$ and $G \setminus \{v\}$. 
(Notice that this also proves the second part of the claim.)
If on the other hand $C$ contains one of the two vertices, say $v$, then $C$ is an SCC of $G \setminus \{u\}$ since all vertices of $C$ and the edges among them remain in $G \setminus \{u\}$.
We have proven both directions of the necessary and sufficient condition.
The lemma follows.
\end{proof}

Notice that Lemma \ref{lem:SCC-partition} implies that the SCC of $v$ in $G\setminus e$ is an SCC in $G \setminus \{u\}$, while the SCC of $u$ in $ G \setminus e$ is an SCC in $G \setminus \{v\}$.

As suggested by Lemma \ref{lem:SCC-partition}, we can get the SCCs of $G \setminus e$ as follows.
Let $C$ be the SCC in $G\setminus \{u\}$ that contains $v$.
Then $C$ is an SCC of $G\setminus e$.
All SCCs in $G \setminus \{v\}$ that do not contain any vertex in $C$ are SCCs of $G\setminus e$ (in particular, we need to test only for one arbitrary vertex since by Lemma \ref{lem:SCC-partition} the SCCs of $G \setminus \{v\}$ and those of $G \setminus \{u\}$ are either nested or disjoint).
Recall that we denote by $A[u,v]$ the index of the SCC of $G \setminus \{v\}$ that contains $u$.

In order to dynamize the above algorithm we handle edge deletions as follows.
Assume that an SCC $C$ of $G \setminus \{v\}$, for some $v$, breaks into several SCCs $C_1, \dots, C_l$.
Without loss of generality, let $C_l$ be the largest SCC that is created.
Recall from Section \ref{sec:maintain-SCC-IDs}, only the SCCs $C_1,\dots, C_{l-1}$ are listed since we cannot afford to list them all.
Therefore, for each strong bridge $e=(u,v)$ or $e=(v,u)$ we test for an arbitrary vertex $z\in C_i$, for each $1 \leq i \leq l-1$, whether $A[v,u] = A[z,u]$. 
If the above condition does holds we remove from $C$ the vertices of $C_i$ and we add $C_i$ as an SCC of $G \setminus e$, otherwise we remove from $C$ the vertices of $C_i$ but we do not add $C_i$ as an SCC of $G\setminus e$.
We note that the tests $A[v,u] = A[z,u]$ refer to the indexes of the SCCs after all the necessary updates after the corresponding edge deletion.

Now we show the correctness of the above algorithm.
The SCCs of $G\setminus e$ that contain neither $u$ nor $v$ are SCCs of $G\setminus \{u\}$.
As implied by Lemma \ref{lem:SCC-partition}, the SCC of $G\setminus e$ containing $v$ appears as an SCC in $G\setminus \{u\}$ (since there is no SCC in $G\setminus \{v\}$ containing $v$).
Thus, if an SCC of $G\setminus e$ does not contain $u$, then it appears as an SCC of $G\setminus \{u\}$.
Since the SCC of $G\setminus e$ containing $u$ is an SCC of $G\setminus \{v\}$, we can test whether a vertex $z$ is in the same SCC with $u$ in $G \setminus e$ we simply can test if it is the same SCC with $u$ in $G\setminus \{v\}$.
Therefore, the SCCs of $G\setminus e$ that do not contain $u$ are correctly updated.
The SCCs of $G\setminus e$ that do not contain $v$ are correctly updated when we examine the SCCs of $G\setminus \{v\}$.
Hence, we correctly update all the SCCs of $G\setminus e$.

\begin{lemma}
\label{lem:num-of-strong-bridges}
Let $G$ be a digraph.
Throughout any sequence of edge deletions, at most $2(n-1)$ strong bridges can appear in $G$.
\end{lemma}
\begin{proof}
Once an edge becomes a strong bridge, it remains a strong bridge until its endpoints are separated in different SCCs.
Therefore, it suffices to bound the number of strong bridges whose endpoints end up into different SCCs after some SCC breaks into smaller SCCs.
Assume an SCC $C$ breaks into $k$ new SCCs $C_1, \dots, C_k$.
Consider the SCC $C$, before it breaks, and an arbitrary vertex $s\in C$.
Let any path $P$ in $G[C]$ from s to any vertex in $C_i$.
Without loss of generality, assume that only the last vertex on the path belongs to $C_i$.
Let $e=(u,v)$ be the last edge on $P$.
Then, there is no other edge $e_2\in C_j \times C_i, e_2 \not= e_1$ in $G[C]$ that disconnects a vertex $w\in C_i$ from $s$, since the path $P$ followed by any path from $v$ to $w$ in $G[C_i]$ avoids $e_2$.
Therefore, for each SCC there is at most one incoming edge disconnecting its vertices from $s$.
If we apply the same argument on the reverse graph we get that each SCC has at most one outgoing edge disconnecting its vertices from $s$.
Since, by Lemma \ref{lem:strong-bridges-on-dominator-trees} all the strong bridges of a strongly connected graph are the edges that disconnected vertices from an arbitrary vertex $s$, if follows that there are at most $2(k-1)$ strong bridges whose endpoints lie in different SCCs $C_i, C_j$.
Therefore, we charge at most $2k$ strong bridges every time the number of SCCs of the graph increases by $k$. 
The lemma follows from the fact that the number of SCCs can only increase and also that after all edges are removed from the graph there are $n$ SCCs left.
\end{proof}

Now we bound the total running time of the algorithm.
By Lemma \ref{lem:num-of-strong-bridges}, at most $O(n)$ strong bridges appear in the graph throughout the execution of the algorithm.
We show that for each strong bridge $(u,v)$, our algorithm spends at most $O(n)$ time.
Indeed, whenever an SCC in $G\setminus \{u\}$ or in $G\setminus \{v\}$ breaks we execute $O(k+k')$ constant time queries, where $k$ and $k'$ are the number of the resulting SCCs after an SCC breaks in $G\setminus \{u\}$ and in $G \setminus \{v\}$, respectively.
Consider the number of queries that are executed whenever SCCs of $G \setminus \{u\}$ break.
After an SCC breaks and $k$ SCCs are created, the number of SCCs increases by $k-1$.
Since the number of SCCs can only increase, and there can be at most $O(n)$ cases where an SCC breaks, it follows that we execute at most $O(n)$ queries for the strong bridge $(u,v)$ when SCCs of $G \setminus \{u\}$ breaks.
The same holds when SCCs of $G \setminus \{v\}$ break.
Moreover, whenever an SCC of $G \setminus \{u\}$, for some vertex $u$, breaks into several SCCs then we spend time proportional to the size of all the newly created SCCs, excluding the largest one, to remove the vertices from the SCC that broke to new SCCs.
With an analysis similar to Section \ref{sec:maintain-SCC-IDs} we can show that we move each vertex $w$ at most $\log n$ times for each strong bridge, which results in $O(n^2 \log n)$ total time spent in the moving of the vertices.
Thus, we spend in total $O(n^2 )$ constant time queries throughout the course of the algorithm maintaining the SCCs of $G \setminus e$, for each strong bridge $e$.

We can moreover build a data structure that can report each SCC in time proportional to its size, once its ID is specified. 
This can be done similarly to Section  \ref{sec:maintain-SCC-IDs} by using doubly linked lists.
We also assume that we maintain the list of IDs of the SCCs in $G \setminus e$, for each edge $e$.

\subsection{Answering decrementally strong connectivity queries under edge failures}

In this section we show how to answer various strong connectivity queries under single edge failures in asymptotically optimal time while we maintain a directed graph $G=(V,E)$ decrementally.
More specifically, under any sequence of edge deletions, we consider answering the following types of queries:
\begin{itemize}
	\item[\it{(i)}] Report the total number of SCCs in $G\setminus e$, for a query edge $e \in E$.
	\vspace{-.1cm}
	\item[\it{(ii)}] Report the size of the largest and of the smallest SCC in $G\setminus e$, for a query edge $e\in E$.
	\item[\it{(iii)}] Report all the SCCs of $G\setminus e$, for a query edge $e\in E$.
	\item[\it{(iv)}] Test if two query vertices $w$ and $z$ are strongly connected in $G\setminus e$, for a query edge $e\in E$.
	\item[\it{(v)}] For query vertices $w$ and $z$ that are strongly connected in $G$,
	report all edges $e$ such that $w$ and $z$ are not strongly connected in $G\setminus e$ anymore.
\end{itemize}

In Section \ref{sec:queries-vertex-failures} we shown how to answer the same queries in asymptotically optimal time with respect to vertex failures.
For static strongly connected graphs, it is known that after linear time preprocessing one can answer all of the above queries in optimal time \cite{GIP:SC}.
Before proving the main result of this section we need two supporting lemmas. 
By the definition of strong bridges, an edge $e=(d(v),v)$, for some $v$, is a strong bridge if all paths from $s$ to $v$ contain $e$.
Since all paths from $s$ to a vertex in $D(v)$ contain $v$, it follows that all paths from $s$ to vertices in $D(v)$ contain $e$.
The first supporting lemma shows that this is not only true for $s$ but also for all vertices that are not in $D(v)$.

\begin{lemma}\emph{(\cite{2ECB})}
	\label{lemma:paths-through-strong-bridge}
	Let $G$ be a strongly connected digraph, let $D$ be its dominator tree rooted at an arbitrary start vertex $s$, and let $e=(u,v)$ be a strong bridge of $G$ such that $d(v)=u$. 
	Then there is a path from any vertex $w\notin D(v)$ to $v$ in $G$ that does not contain any vertex in $D(v)$. 
	Moreover, all simple paths in $G$ from $w$ to a vertex in $D(v)$ must contain $e$.
\end{lemma}

\begin{lemma}[\cite{GIP:SC}]
	\label{edges-separators-ancestors}
	Let $e=(u,v)$ be a strong bridge that is a separating edge for vertices $w$ and $z$. Then $e$ must appear in at least one of the paths $D[s,w]$, $D[s,z]$, $D^{\R}[s,w]$, and $D^{\R}[s,z]$.
\end{lemma}

Now we are ready to prove the main result of this section.

\begin{lemma}
	We can maintain a digraph $G$ decrementally in $O(mn \log n)$ total expected update time against an oblivious adversary, using $O(n^2 \log n)$ space, where $m$ is the number of edges in the initial graph and $n$ is the number of vertices, and between any two edge deletions  answer in asymptotically optimal time the following type of queries under edge failures:
	\begin{itemize}
		\item[\it{(i)}] Report in $O(1)$ time the total number of SCCs in $G\setminus e$, for a query edge $e \in E$.
		\vspace{-.1cm}
		\item[\it{(ii)}] Report in $O(1)$ time the size of the largest and of the smallest SCC in $G\setminus e$, for a query edge $e \in E$.
		\item[\it{(iii)}] Report in $O(n)$ worst-case time all the SCCs of $G\setminus e$, for a query edge $e \in E$.
		\item[\it{(iv)}] Test in $O(1)$ time if two query vertices $w$ and $z$ are strongly connected in $G\setminus e$, for a query edge $e \in E$.
		\item[\it{(v)}] For query vertices $w$ and $z$ that are strongly connected in $G$,
		report all edges $e$ such that $w$ and $z$ are not strongly connected in $G\setminus e$, in optimal worst-case time, i.e., in time $O(k)$, where $k$ is the number of separating vertices.
		(For $k=0$, the time is $O(1)$).
	\end{itemize}
\end{lemma}
\begin{proof}
	Queries \textit{(i),(iii),(iv)} can be answered by maintaining the labels $A[u,v]$, for each $u,v\in V$, as shown in Section \ref{sec:maintain-SCC-IDs-edge-failures}.
	As we also mentioned in Section \ref{sec:maintain-SCC-IDs-edge-failures}, we can maintain for each $G\setminus e$ a list of its SCCs.
	This list can be easily extended to maintain the size of each SCC.
	In order to have fast access to the minimum and the maximum size of the SCCs, we store the sizes of the SCCs in a min-heap and a max-heap.
	Those heaps can be updated in total time $O(n \log n)$ for each subgraph $G\setminus e$, for each $e$, as follows.
	Whenever an SCC breaks, we add the IDs of the newly created SCCs together with their size into the heaps, and we also update the size of the one that kept the same ID.
	Since at most $n$ SCCs can be created, and moreover there can be at most $n$ cases of a broken SCC, there exist at most $O(n)$ insertions and updates to each heap.
	That means the total time spend on maintaining each heap for each $G\setminus e$ is $O(n \log n)$, which sums up to $O(n^2\log n)$ for all strong bridges $e$ (since the number of strong bridges that appear throughout the algorithm is \ref{lem:num-of-strong-bridges}).
	Given one min-heap and one-max heap, we can answer type \textit{(ii)} queries in constant time.
	
	We are left to show how to answer the queries of type \textit{(v)}.
	For this type of queries we will assume that we maintain the dominator tree $D$ of the graph and the dominator tree $D^{\R}$ of the reverse graph.
	We assume that after each edge deletion we compute the last strong bridge on the path from $s$ to $q$ on $D$, for each vertex $q$, denoted by $\ell(q)$.
	If no such edge exists for some vertex $q$, we let $\ell(q) = \null$
	This can be easily done in $O(n)$ time.

	By Lemma \ref{edges-separators-ancestors}, all separating edges for $w$ and $z$ are either ancestors of $w$ or $z$ in $D$ or ancestors of $w$ or $z$ in $D^{\R}$.
	We only show how to report the separating edges for $w$ and $z$ that are ancestors of $w$ or $z$ in $D$ since the procedure for $D^{\R}$ is completely analogous.
	By Lemma \ref{lemma:paths-through-strong-bridge}, notice that if there exist an edge $e'=(t,q)$ such that $w\in D(q)$ and $z\notin D(q)$ or such that $z\in D(q)$ and $w\notin D(q)$ then $e'$ is a separating edge for $w$ and $z$.
	That means, all strong bridges on the path in $D$ from $\NCA_D(w,z)$ (the nearest common ancestor of $ w $ and $ z $ in $ D $) to $d(w)$ and from $\NCA_D(w,z)$ to $d(z)$ are separating edges for $w$ and $z$.
	We can simply report all strong bridges on the path in $D$ from $\NCA_D(w,z)$ to $w$ and on the path in $D$ from $\NCA_D(w,z)$ to $z$.
	These strong bridges can be reported in output sensitive time, plus $O(1)$ time, as follows.
	We first check whether for $\ell(w)=(x,y)$ if holds that $x\in D(\NCA_D(w,z))$, and if yes we report $\ell(w)$, otherwise we stop.
	We continue in the same way by testing $\ell(x)$, etc.
	We repeat the same process to report all strong bridges on the path from $\NCA_D(w,z)$ to $z$.
	
	Next, we show that all edges $e'=(t,q)$ such that $t,q \notin D(\NCA_D(w,z))$ that separate $w$ and $z$ appear as consecutive strong bridges on a path in $D$.
	More specifically, there is a vertex $t\in V\cup \{\emptyset\} \setminus\{D(\NCA_D(w,z))\}$, such that all strong bridges on the path from $t$ to $\NCA_D(w,z)$ in $D$ are separating edges for $w$ and $z$.
	Let $(t,q)$ be the first edge on the path from $s$ to $\NCA_D(w,z)$ on $D$ that separates $w$ and $z$.
	If there is no such edge $(t,q)$, then none of the edges on the path from $s$ to $\NCA_D(w,z)$ is a separating edge for $w$ and $z$ and we are done.
	We can verify the existence of such edge by testing if $\ell(\NCA_D(w,z)) \not= \emptyset$ and $\ell(\NCA_D(w,z))$ is a separating edge for $w$ and $z$ (the test can be executed in constant time as one type \textit{(iv)} query).
	Assume now that such a separating edge $(t,q)$ exists.
	By Lemma \ref{lemma:paths-through-strong-bridge}, either all paths from $w$ to $z$ contain $(t,q)$ or all paths from $z$ to $w$ contain $(t,q)$.
	Assume, without loss of generality, that all paths from $w$ to $z$ contain $(t,q)$.
	Since $q\notin D(\NCA_D(w,z))$, all paths from $q$ to $z$ contain all strong bridges on the path from $q$ to $z$ in $D$ (including all strong bridges on the path from $q$ to $\NCA_D(w,z)$).
	This allows us to efficiently identify and report all separating edges $(t,q)$ such that $t,q\notin D(\NCA_D(w,z))$ for $w$ and $z$ as follows.
	If there exists one such separating edge, we start testing the edge on the path from $\NCA_D(w,z)$ to $s$ in $D$ (by testing the strong bridges provided by following the pointers $\ell$ as done previously), reporting all strong bridges that are separating edges for $w$ and $z$, and once we find an edge that is not separating $w$ and $z$ we stop (as we proved, there are not further edge on the path from $s$ to $\NCA_D(w,z)$).
%

	Notice that we only spend time proportional to the edges that we report, the computation of $\NCA_D(w,z)$, and only a single type \textit{(iv)} query that does not report an edge.
	We also spend the analogous time on the dominator tree $D^{\R}$ of the reverse graph.
	We only need to be careful not to report the same vertex twice, which can be trivially implemented within the claimed time bound.
	The Lemma follows.
\end{proof}

\subsection{Maintaining decrementally the $2$-edge-connected components}
\label{sec:2ECBs}

In this section we show how to maintain the $2$-edge-connected components of directed graph.
By definition, two vertices $w$ and $z$ are $2$-edge-connected if and only if there is no edge $e$ such that $w$ and $z$ are not strongly connected in $G\setminus e$.
Therefore, a simple-minded algorithm for computing the $2$-edge-connected components is the following.
We start with the trivial partition $\mathcal{P}$ of the vertices that is equal to the set of SCCs of the graph.
For every strong bridge $e$, we compute the SCCs $C_1,\dots, C_k$ of $G\setminus e$ and we refine the maintained partition $\mathcal{P}$ according to the partition induced by the SCCs $C_1, \dots, C_k$.
After we execute all of the refinements on $\mathcal{P}$ two vertices are in the same set if and only if we did not find an edge that separates them, which is exactly the definition of $2$-edge-connected components.

Our algorithm is a dynamic version of the aforementioned simple-minded algorithm.
That is, we are executing the refined operations after every edge deletion and only if it is necessary.
More specifically, we maintain decrementally the SCCs of $G \setminus e$ for each strong bridge $e$.
Whenever we identify that a $2$-vertex-connected component $B$ contains vertices from different SCCs in $G\setminus e$, for some strong bridge $e$, then we refine the $2$-edge-connected components according to these SCCs.
Assume that a $2$-edge-connected component $C$ breaks after an edge deletion.
That is, there is a strong bridge $e$ such that vertices of $B$ lie in different SCCs in $G\setminus e$.
Let $C_1,C_2,\dots, C_k$ be the SCCs in $G\setminus e$, then we replace $B$ by $\{B\cap C_1\}, \dots, \{B \cap C_k\}$. 
Notice that we can afford to spend up to $O(m)$ time every time we detect that a $2$-edge-connected component should be refined, as this can happen at most $O(n)$ times (follows from the fact that each time at least two vertices stop being $2$-edge-connected).
However, these refinements can be easily executed in $O(n)$ time, and therefore we spend total time $O(n^2)$ for all refinements throughout the algorithm.

In order to make our algorithm efficient we need to specify how to detect whether two $2$-edge-connected vertices appear in different SCCs in $G\setminus e$, for some edge $e$.
Whenever an SCC $C$ in $G\setminus e$, breaks into several SCCs $C_1,\dots, C_k$ the vertices of all SCCs except of one can be listed in time proportional to their edges as shown in Section \ref{sec:maintain-SCC-IDs-edge-failures}.
Without loss of generality, let $C_1, \dots, C_{k-1}$ be those SCCs.
For SCC $C_i$, for $1\leq i \leq k-1$, we examine whether the $2$-edge-connected components containing subsets of vertices of $C_i$ are entirely contained in $C_i$.
This can be easily done in time proportional to $|C_i|$, by simply testing the ID of their $2$-edge-connected components.
Notice that we do no need to examine the vertices in $C_k$ since if we do not find a $2$-edge-connected pair that is disconnected in $G\setminus e$ by the searches in $C_i, 1\leq i\leq k-1$, then there is no such pair, as we now explain.
Assume that there is a pair of $2$-edge-connected vertices $w,z$ such that $w\in C_k, z\notin C_k$. 
By the fact that $w$ and $z$ were $2$-edge-connected before the edge deletion it follows that $z\in C$, and therefore $z\in C_i, 1\leq i\leq k-1$, in which case we would find this by searching in $C_i$.

If we detect a $2$-edge-connected component whose vertices lie in different SCCs in $G\setminus e$, for some strong bridge $e$, then we execute the refinement phase in $O(n)$ time.
Notice that all the necessary tests that we described above can be executed in time proportional to the number of vertices of a broken SCC in $G\setminus \{v\}$, for some $v$, that are not contained in the largest resulting SCC.
Each vertex can appear at most $\log n$ times in an SCC of $G\setminus e$, for some strong bridge $e$, that is not the largest after a big SCC breaks. 
That means we spend $O(n \log n)$ time for each graph $G\setminus e$ on the aforementioned queries, for some strong bridge $e$, and therefore $O(n^2 \log n)$ in total.

\begin{lemma}
	\label{lemma:2ECCs-running-time}
	The $2$-edge-connected components of a directed graph $G$  can be maintained decrementally in $O(mn \log n)$ total expected time against an oblivious adversary, using $O(n^2 \log n)$ space, where $m$ is the number of edges in the initial graph and $n$ is the number of vertices.
\end{lemma}

\subsection{Maintaining decrementally the maximal $2$-edge-connected subgraphs}

A strongly connected graph $G=(V,E)$ is $2$-edge-connected if for each $e\in E$ it holds that $G\setminus e$ remains strongly connected.
In this section we show how to maintain the maximal $2$-edge-connected subgraphs of a directed graph decrementally.
The maximal $2$-vertex-connected subgraphs induce a partition of the vertices of the graph.
A simple-minded algorithm for computing the $2$-edge-connected subgraphs removes iteratively one strong bridges from each SCC until there are no more strong bridges in any SCC.
Clearly the remaining SCCs of the resulting graph are $2$-edge-connected subgraphs since they do not contain any strong bridges.
Moreover, all maximal $2$-edge-connected subgraphs of the initial graph remain intact since, no edge can separate its vertices in different SCCs (including the edges inside the $2$-edge-connected subgraph).

We now present our decremental algorithm for maintaining the maximal $2$-edge-connected subgraphs of $G$.
Our algorithms is a simple extension of the simple-minded static algorithm for computing the maximal $2$-edge-connected subgraphs.
That is, we maintain the graph $G'$ resulting after removing iteratively all strong bridges.
Notice that the maximal $2$-edge-connected subgraphs can only be further partitioned, so we can ignore all the edges between the different SCCs of  $G'$.
In order to do so, we simply remove them from $G'$, once we discover such edges.
Therefore after an edge deletion, we only need to detect whether some edge deletion introduces new strong bridges inside the SCCs of $G'$.
If such a strong bridge is discovered, we simply remove it from $G'$ and also remove all edges among the different resulting SCCs.
Throughout this process we continue searching for new strong bridges in the maintained SCCs and we repeat this process.

Now we show how to implement the aforementioned decremental algorithm efficiently.
We maintain the SCCs in $G\setminus e$ for each strong bridge $e$, as shown in Section \ref{sec:maintain-SCC-IDs-edge-failures}.
Whenever an SCC in $G'\setminus e$, for some strong bridge $e$, breaks into several SCCs $C_1,\dots, C_k$ after an edge deletion, we collect all edges between different SCCs and remove them from $G'$ (we assume that we keep all edges that are removed from this process in a global set data structure $L$, and they are executed as normal edge deletion from our data structure).
The above procedure maintains $G'$ under any sequence of edge deletions, since for each strong bridge $e$ we only remove edges between different SCCs in $G' \setminus e$ and when the additional edge deletions end there are no edges between different SCCs in $G' \setminus e$.

Now we bound the running time of the algorithm.
The time for handling all edge deletions in the maintained graph $G'$ is bounded by $O(mn \log n)$, and it uses $O(n^2 \log n)$ space, by Lemma \ref{lem:mn-space-efficient}.
We maintain the SCCs in $G'\setminus e$ for each strong bridge $e$, in total $O(n^2 \log n)$ time for all strong bridges, as shown in Section \ref{sec:maintain-SCC-IDs-edge-failures}.
We only need to bound the time we spend to collect all edges among different resulting SCCs after some SCC in $G' \setminus e$ breaks, for any strong bridge $e$.
We do that as follows.
Whenever an SCC $C$ breaks into several SCCs $C_1,\dots, C_k$ in $G' \setminus e$, for some strong bridge $e$, we only need to identify all the edges among $C_1, \dots, C_k$, since the algorithm previously should have removed all edges from $C$ to other SCCs in $G'\setminus e$.
Without loss of generality, let $C_k$ be the largest SCC among $C_1,\dots, C_k$.
For each $C_i$, $1\leq i \leq k-1$ we iterate over all the unmarked edges incident to their vertices and test whether their endpoints are in different SCCs in $G'\setminus \{v\}$, and if yes we mark them and insert them into the global set data structure $L$ in order to delete them from the graph.
Notice that whenever the algorithm iterates over the incident edges of a vertex, that vertex is contained in an SCC in $G' \setminus e$ that is at most half the size of its previous SCC.
That means each vertex will be listed at most $\log n$ times.
Therefore, for each strong bridge $e$, we consider the edges of each vertex at most $\log n$ times, and therefore at most $n\log n$ times in total, for all strong bridges.
Hence we spend at most $O(mn\log n)$ to collect all edges among different resulting SCCs after some SCC in $G' \setminus e$ breaks, for any strong bridge $e$.
Thus, we have the following lemma.

\begin{lemma}
	\label{lem:2ECSs-running-time}
	The maximal $2$-edge-connected subgraphs of a directed graph $G$  can be maintained decrementally in $O(mn \log n)$ total expected update time against an oblivious adversary, using $O(n^2 \log n)$ space, where $m$ is the number of edges in the initial graph and $n$ is the number of vertices.
\end{lemma}

\section{Decremental Dominators in Reducible Graphs}
\label{sec:reducible}

In this section we give a specialized solution for maintaining the dominator tree under edge deletions in reducible flow graphs.
The algorithm has a total update time of $ O (m n) $ and uses space $ O (m + n) $.

A \emph{reducible} flow graph~\cite{rdflow:hu74,reducibility:jcss:tarjan} is one in which every strongly connected subgraph $S$ has a single \emph{entry} vertex $v$ such that every path from $s$ to a vertex in $S$ contains $v$.  There are many equivalent characterizations of reducible flow graphs~\cite{reducibility:jcss:tarjan}, and there are algorithms to test reducibility in near-linear~\cite{reducibility:jcss:tarjan} and truly linear~\cite{dominators:bgkrtw} time. A flow graph is reducible if and only if it becomes acyclic when every edge $(v, w)$ such that $w$ dominates $v$ is deleted~\cite{reducibility:jcss:tarjan}.  We refer to such an edge as a \emph{back edge}.
Deletion of such edges does not change the dominator tree, since no such edge can be on a simple path from $s$.
Deleting such edges thus reduces the problem of computing dominators on a reducible flow graph to the same problem on an acyclic graph.
Such a graph has a topological order (a total order such that if $(x, y)$ is an edge, $x$ is ordered before $y$) \cite{fundamental-algorithms:Knuth}.

It is well-known that the dominator tree $D$ of an acyclic flow graph $G$ can be computed by the following simple algorithm, which builds $D$ incrementally~\cite{rdflow:hu74}.
Fix a topological order of $G$ (for the vertices reachable from $s$).
Initially $D$ consists of only its root $s$. We process the vertices in topological order, and for each vertex $v$ we compute the nearest common ancestor $u$ of $\mathit{In}(v)$ in $D$.
Then we set $d(v) \leftarrow u$.

\subsection{Preliminary Observations}

\subsubsection{The Parent and Sibling Properties}
\label{sec:parent-sibling-properties}

Let $T$ be a rooted tree whose vertex set $V(T)$ consists of the vertices reachable from $s$.
Tree $T$ has the \emph{parent property} if for all $(v, w) \in E$ with $v$ and $w$ reachable, $v$ is a descendant of $t(w)$ in $T$.
Tree~$T$ has the \emph{sibling property} if $v$ does not dominate $w$ for all siblings $v$ and $w$ in $T$.
The parent and sibling properties are necessary and sufficient for a tree to be the dominator tree.

\begin{theorem}\emph{(\cite{DomCert:TALG})}
\label{theorem:parent-sibling}
A tree $T$ has the parent and sibling properties if and only if $T = D$.
\end{theorem}

\subsubsection{Derived Edges and Derived Graphs}
\label{sec:derived}

Derived graphs, first defined in \cite{path:tarjan81}, reduce the problem of finding dominators to the case of a flat dominator tree.
By the parent property of $D$, if $(v, w)$ is an edge of $G$, the parent $d(w)$ of $w$ is an ancestor of $v$ in $D$.
Let $(v,w)$ be an edge of $G$, with $w$ not an ancestor of $v$ in $D$.
Then, the \emph{derived edge} of $(v, w)$ is the edge $(\overline{v}, w)$, where $\overline{v} = v$ if $v = d(w)$, $\overline{v}$ is the sibling of $w$ that is an ancestor of $v$ if $v \not= d(w)$. If $w$ is an ancestor of $v$ in $D$, then the derived edge of $(v, w)$ is null. Note that a derived edge $(\overline{v}, w)$ may not be an original edge of $G$.
For any vertex $w \in V$ such that $C(w) \not= \emptyset$, we define the \emph{derived flow graph of $w$}, denoted by $G_w = (V_w, E_w, w)$, as the flow graph with
start vertex $w$,
vertex set $V_w = C(w) \cup \{ w \}$, and edge set $E_w = \{ (\overline{u}, v) \ | \ v \in V_w \mbox{ and } (\overline{u}, v) \mbox{ is the non-null derived edge of some edge in } E \}$.
By definition, $G_w$ has flat dominator tree, that is, $w$ is the only proper dominator of any vertex $v \in V_w \setminus w$.

\begin{lemma}\emph{(\cite{DomCert:TALG})}
\label{lemma:derived}
Given the dominator tree $D$ of a flow graph $G=(V,E,s)$ and a list of edges $S \subseteq E$, we can compute the derived edges of $S$ in $O(|V|+|S|)$ time.
\end{lemma}

\subsubsection{Affected Vertices}
\label{sec:affected}

Now consider the effect that a single edge deletion has on the dominator tree $D$. Let $(x,y)$ be the deleted edge. We let $G'$ and $D'$ denote the flow graph and its dominator tree after the update. Similarly, for any function $f$ on $V$, we let $f'$ be the function after the update. In particular, $d'(v)$ denotes the parent of $v$ in $D'$.
By definition, $D' \not= D$ only if $x$ is reachable before the update. We say that a vertex $v$ is \emph{affected} by the update if $d'(v) \not= d(v)$. (Note that we can have $\mathit{Dom}'(v) \not= \mathit{Dom}(v)$ even if $v$ is not affected.)
If $v$ is affected then $d'(v)$ does not dominate $v$ in $G$.

Suppose that $x$ is reachable and $y$ remains reachable after the deletion of $(x,y)$. The deletion of an edge does not violate the parent property of the dominator tree but may violoate the sibling property. Since the effect of an edge deletion is the reverse of an edge insertion, \cite[Lemma 1]{dyndom:2012} and \cite[Lemma 3.8]{IncLowHigh:arXiv} give the following result:

\begin{lemma}
\label{lemma:deletion-child}
Suppose $x$ is reachable and $y$ does not becomes unreachable after the deletion of $(x,y)$. Then the following statements hold:
\begin{itemize}
\item[(a)] A vertex $v$ is affected only if $d(v)=d(y)$ and there is a path $\pi_{yv}$ from $y$ to $v$ such that $\mathit{depth}(d(v)) < \mathit{depth}(w)$ for all $w \in \pi_{yv}$.
\item[(b)] All affected vertices become descendants in $D'$ of a child $c$ of $d(y)$.
\item[(c)] After the deletion, each affected vertex $v$ becomes a child of a vertex on $D'[c,y]$.
\end{itemize}
\end{lemma}

\begin{figure*}
\begin{center}
\includegraphics[width=\textwidth, trim = 0mm 0mm 0mm 80mm, clip]{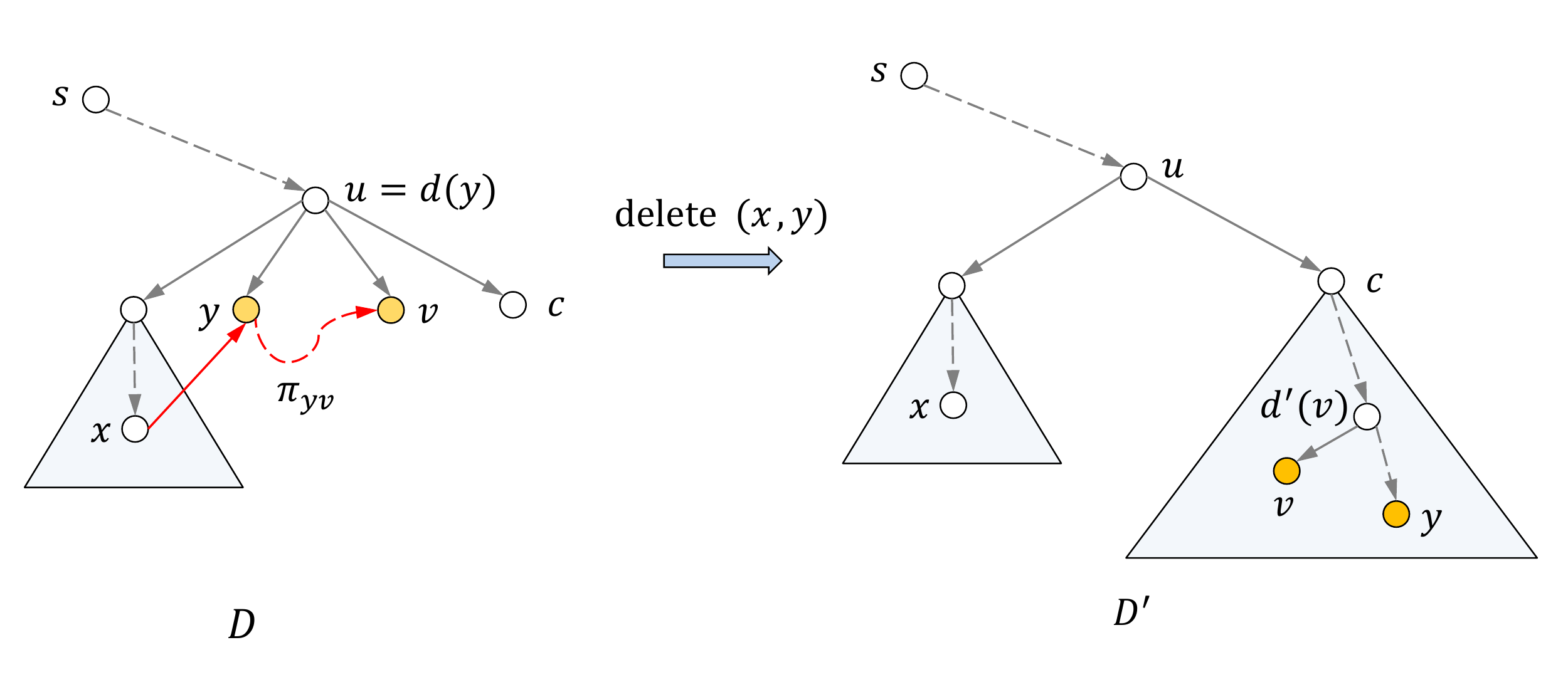}
\caption{\label{fig:deletion} Illustration of Lemma \ref{lemma:deletion-child}.}
\end{center}
\end{figure*}

See Figure \ref{fig:deletion}. We refer to $D'[c,y]$ as the \emph{critical path} of the deletion. Notice that the above lemma only provides a necessary condition for a vertex to be affected.

\subsection{Decremental Algorithm}

Let $G=(V,E,s)$ be the input reducible flow graph. Before we begin executing our decremental algorithm, we compute the dominator tree $D$ of $G$ and delete its back edges.
Henceforth, we assume that $G$ is acyclic.
Throughout the execution of our algorithm we will need to test the ancestor-descendant relation between pairs of vertices in $D$, by applying a simple $O(1)$-time
test~\cite{domin:tarjan}.
Specifically, after each edge update, we perform a dfs traversal of $D$ where we number the vertices from $1$ to $n$ in preorder and compute the number of descendants of each vertex $v$.
We denote these numbers by $\mathit{preorder}(v)$ and $\mathit{size}(v)$, respectively.
Then $v$ is a descendant of $u$ if and only if $\mathit{preorder}(u) \le \mathit{preorder}(v) < \mathit{preorder}(u) + \mathit{size}(u)$.

The next lemma follows from \cite[Lemma 4.1]{DomCert:TALG}.

\begin{lemma}
\label{lemma:reducible-condition}
Suppose $x$ is reachable and $y$ does not become unreachable after the deletion of $(x,y)$.
Then $y$ is affected if and only if $(d(y),y)$ is not an edge of $G \setminus (x,y)$ and all edges $(v,y) \in E \setminus (x,y)$ correspond to the same derived edge $(\overline{v},y)=(c,y)$ of $G$.
\end{lemma}

Our goal is to apply Lemma \ref{lemma:reducible-condition} in order to locate the affected vertices in some topological order of $G$.
For each vertex $v$ we maintain a count $\mathit{InSiblings}(v)$ which corresponds to the number of distinct siblings $w$ of $v$ such that $(w,v)$ is a derived edge.
We also maintain the lists $\mathit{DerivedOut}(v)$ of the derived edges $(v,u)$ leaving each vertex $v$.
As we locate each affected vertex, we find its new parent in the dominator tree and update the counts $\mathit{InSiblings}$ for the siblings of $v$.
We compute the updated $\mathit{InSiblings}$ counts and $\mathit{DerivedOut}$ lists in a postprocessing step.

Let $(x,y)$ be the deleted edge. The first step is to test if $y$ is affected after the deletion, as suggested by Lemma \ref{lemma:reducible-condition}.
Specifically, we compute the nearest common ancestor $z$ of all vertices in $\mathit{In}(y)$. 
From Lemma \ref{lemma:reducible-condition} we have that $y$ is affected if and only if $z \not= d(y)$. In this case, by Lemma \ref{lemma:deletion-child}, $z$ is a descendant of a sibling $c$ of $y$ in $D$. Note that we can locate $z$ and $c$ in $O(n)$ time, using the parent function $d$.

\begin{corollary}
\label{corollary:deletion-child-acyclic}
Let $G$ be an acyclic flow graph. Let $(x,y)$ be an edge of $G$ that is not a bridge such that $x$ is reachable.
Then, after the deletion of $(x,y)$, no vertex on $D'[c,d'(y)]$ is affected.
\end{corollary}
\begin{proof}
The lemma is true if $y$ is not affected, so suppose that $y$ is affected.
Suppose, for contradiction, that there is an affected vertex $v$ on $D'[c,d'(y)]$. Then $v$ is a dominator of $y$ in $G'$, so all paths from $s$ to $y$ in $G'$ contain $v$.
Let $\pi_{sy}$ be such a path, and let $\pi_{vy}$ be the part of $\pi_{sy}$ from $v$ to $y$. Since $G'$ is a subgraph of $G$, path $\pi_{vy}$ also exists in $G$.
From the fact that $v$ is affected and from Lemma \ref{lemma:deletion-child}, we have that $G$ also contains a path $\pi_{yv}$ from $y$ to $v$.
But then $G$ is not acyclic, a contradiction.
\end{proof}

Algorithm \textsf{DeleteEdge} gives the outline of our algorithm to update the dominator tree after an edge deletion.
To identify the affected vertices, we update the counters $\mathit{InSiblings}$, using Procedure~\ref{Procedure:UpdateInSiblings}, and apply Lemma \ref{lemma:reducible-condition}.
We store the affected vertices in a queue $Q$ and process them as they are extracted from $Q$. To process an affected vertex $w \not= y$, we first need to locate the position of $d'(w)$ on
the critical path $D'[c,d'(y)]$. This is handled by Procedure~\ref{Procedure:LocateNewParent}.

\begin{algorithm}
\LinesNumbered
\DontPrintSemicolon
\KwIn{Flow graph $G=(V,E,s)$, its dominator tree $D$, arrays $\mathit{preorder}$ and $\mathit{size}$, and an edge $e=(x,y)$.}
\KwOut{Flow graph $G' = (V, E \setminus (x,y), s)$, its dominator tree $D'$, and arrays $\mathit{preorder}'$ and $\mathit{size}'$.}

Delete $e$ from $G$ to obtain $G'=(V,E',s)$.\;

\lIf{$x$ was unreachable in $G$}{\KwRet $(G', D, \mathit{preorder}, \mathit{size})$}
\ElseIf{$y$ is becomes unreachable in $G'$}{
$(D', \mathit{preorder}', \mathit{size}') \leftarrow \mathsf{Initialize}(G')$\;
\KwRet $(G', D', \mathit{preorder}', \mathit{size}')$
}

Let $\mathit{In}(y)$ be the set of vertices $v$ such that $(v,y)$ is an edge in $G'$.\;

Let $f$ be the child of $d(y)$ that is an ancestor of $x$.\;

\If{there is no vertex $v \in \mathit{In}(y)$ such that $v \in D(f)$}{
    Set  $\mathit{InSiblings}(y) \leftarrow InSiblings(y)-1$.\;
    Set  $\mathit{DerivedOut}(f) \leftarrow \mathit{DerivedOut}(f) \setminus y$.
}

\lIf{$(d(y),y) \in E'$ or $\mathit{InSiblings}(y) \ge 2$}{\KwRet $(G', D, \mathit{preorder}, \mathit{size})$}

Compute the nearest common ancestor $z$ of $\mathit{In}(y)$ in $D$.\;

\lIf{$z=d(y)$}{
    \KwRet $(G, D, \mathit{preorder}, \mathit{size})$
}

Let $c$ be the child of $d(y)$ that is an ancestor of $y$ in $D$.\;

Set $d'(y) \leftarrow z$.\;

Execute \ref{Procedure:UpdateInSiblings}$(y)$.\;


\While{$Q$ is not empty} {
    Extract a vertex $w$ from $Q$.\;

    \lForAll{$v \in D(w)$}{set $\mathit{AffectedAncestor}(v) \leftarrow w$.}

    Execute \ref{Procedure:LocateNewParent}$(w)$ and \ref{Procedure:UpdateInSiblings}$(w)$.
}

Delete the affected vertices from $\mathit{DerivedOut}(c)$.\;

Let $S$ be the set of all edges entering affected vertices. Compute the derived edges $\overline{S}$ of $S$.\;

Compute $\mathit{InSiblings}(w)$ for all affected vertices $w$.\;

Compute $\mathit{DerivedOut}(v)$ for all vertices $v$ such that $(v,w) \in \overline{S}$.\;

Make a dfs traversal of $D'$ to compute the updated arrays $\mathit{preorder}'$ and $\mathit{size}'$. \;

\KwRet $(G', D', \mathit{preorder}', \mathit{size}')$

\caption{\textsf{DeleteEdge}$(G, \mathit{preorder}, \mathit{size}, e)$}
\end{algorithm}

Let $\mathit{degree}_0(v)$ denote the initial degree (indegree and outdegree) of a vertex $v \in V$.
We define the potential $\phi(v)$ of a vertex $v$ as $\phi(v) = \mathit{depth}(v) \cdot \mathit{degree}_0(v)$ if $v$ is reachable, and $\phi(v)=n \cdot \mathit{degree}_0(v)$ otherwise.
The flow graph potential $\Phi$ is the sum of the vertex potentials.
Note that vertex potentials are nondecreasing. Also, $n-1 \le \Phi < 2nm$.
For a set of vertices $S$, we let $\mathit{degree}_0(S) = \sum_{v \in S}\mathit{degree}_0(v)$ and $\Phi(S) = \sum_{v \in S}\phi(v)$.

After we have found a new affected vertex $w$, the next step is to update the $\mathit{InSiblings}$ counters for the siblings $u$ of $w$ in $D$ that have an entering edge from $D'(w)$. Since we discover the affected vertices in topological order, none of these siblings of $w$ has been inserted into $Q$ yet. Moreover, by Lemma \ref{lemma:deletion-child} and Corollary \ref{corollary:deletion-child-acyclic}, no descendant of $w$ in $D$ is affected, so $D'(w)=D(w)$. So the siblings of $w$ that have an entering edge from $D'(w)$ are precisely those in $\mathit{DerivedOut}(w)$.
Then we can update the counters as shown in Procedure~\ref{Procedure:UpdateInSiblings}.

\begin{procedure}
\caption{UpdateInSiblings($w$)}
\label{Procedure:UpdateInSiblings}
\DontPrintSemicolon

\ForEach{vertex $q \in \mathit{DerivedOut}(w)$}
{
    \If{$q \in \mathit{DerivedOut}(c)$}{
        set $\mathit{InSiblings}(q) \leftarrow \mathit{InSiblings}(q)-1$\;
        \lIf{$\mathit{InSiblings}(q) = 1$ and $d(q) \not\in \mathit{In}(u)$}{insert $q$ into $Q$}
    }
    \Else{
        $\mathit{DerivedOut}(c) \leftarrow \mathit{DerivedOut}(c) \cup q$
    }
}
Set  $\mathit{DerivedOut}(w) \leftarrow \emptyset$.\;
\end{procedure}

We say that a vertex is \emph{scanned} if it is visited in line 21 of Algorithm \textsf{DeleteEdge}. Hence, every scanned vertex is a descendant of an affected vertex in $D$. Notice that a vertex $v$ can be scanned at most once per edge deletion, and when this happens an ancestor $w$ of $v$ in $D$ is affected. We maintain this information in a variable $\mathit{AffectedAncestor}(v)$ for each vertex $v$.


Now we describe how to find the new parent of each affected vertex.
Let $w$ be the next affected vertex extracted from $Q$. Lemma \ref{lemma:deletion-child} and Corollary \ref{corollary:deletion-child-acyclic} imply that $d'(w)$ is located on the critical path $D'[c,d'(y)]$.
To locate the position of $d'(w)$, we find the deepest vertex $u$ on the critical path such that setting $d'(w) \leftarrow u$ satisfies the parent property. See
Procedure~\ref{Procedure:LocateNewParent}.
We test the vertices $u \in D'[c,d'(y)]$ in top-down order, because this allows us to charge the cost of these tests to the increase of $\phi(w)$.
(See Theorem \ref{theorem:decremental-running-time} below.)

\begin{procedure}
\caption{LocateNewParent($w$)}
\label{Procedure:LocateNewParent}
\DontPrintSemicolon

\ForEach{vertex $u \in D'(c,d'(y)]$ in top-down order}
{
    \If{there is an edge $(v,w) \in E'$ such that $v \not\in D'(u)$ }{ 
        set $d'(w) \leftarrow d(u)$ and \KwRet\;
    }
}
\end{procedure}

Consider an affected vertex $w$ and a vertex $u$ on the critical path $D'[c,d'(y)]$. Let $(v,w)$ be an edge in $E'$.
To test if $v$ is a descendant of $u$ in $D'$ we consider two cases:
\begin{itemize}
\item If $v$ is not a scanned vertex, then the relative location of $v$ and $u$ has not changed. That is, $v \in D'(u)$ if and only if $v \in D(u)$.
\item If $v$ is a scanned vertex, then it becomes a descendant of a vertex $q$ on $D'[c,d'(y)]$. Let $p=\mathit{AffectedAncestor}(v)$. Then $q=d'(p)$, and since we locate the affected vertices in topological order, $q$ is already known. So we have $v \in D'(u)$ if and only if $q \in D(u)$.
\end{itemize}

Therefore, in both cases we can use the $O(1)$-time test for the ancestor-descendant relation in $D$, using arrays $\mathit{preorder}$ and $\mathit{size}$.

\begin{lemma}
\label{lemma:decremental-correctness}
Algorithm \textsf{DeleteEdge} is correct.
\end{lemma}
\begin{proof}
We will argue that the affected vertices are processed in topological order. The correctness of our algorithm follows from this fact because each execution of Procedure \ref{Procedure:LocateNewParent}$(w)$ sets $d'(w)$ to be the nearest common ancestor of $\mathit{In}(w)$. This means that our algorithm updates the dominator tree by mimicking the incremental construction of the dominator tree~\cite{rdflow:hu74}. 

To prove our claim that the affected vertices are processed in topological order, we first note that by Lemma \ref{lemma:deletion-child} there is a path in $G'$ from $y$ to each affected vertex. So, $y$ is the first affected vertex in topological order.
A vertex $w$ is inserted into $Q$ if and only if $\mathit{InSiblings}(w)=1$ and $(d(w),w)$ is not an edge of $G'$.
Suppose that $w$ is inserted into $Q$.
Since $w \not= y$, $(d(w),w) \not= (x,y)$. Hence, $(d(w),w)$ is not an edge of $G$ and $\mathit{InSiblings}(w)>1$ before the deletion.
Thus $\mathit{InSiblings}(w)$ was decreased by Procedure~\ref{Procedure:UpdateInSiblings}.
Then, the condition in line 3 of that procedure implies that the moment $w$ is inserted into $Q$, all vertices in $\mathit{In}(w)$ have become descendants of $c$.
This proves the claim, so the lemma follows.
\end{proof}

\begin{theorem}
\label{theorem:decremental-running-time}
Algorithm \textsf{DeleteEdge} maintains the dominator tree of a reducible flow graph $G$ with $n$ vertices through a sequence of edge deletions in
$O(mn)$ total time, where $m$ is number of edges in $G$ before all deletions.
\end{theorem}
\begin{proof}
Consider the deletion of an edge $e=(x,y)$. If $e$ is a bridge, then $y$ becomes unreachable in $G'$ and we
compute $D'$ from scratch in $O(m)$ time. Throughout
the whole sequence of deletions, such an event can happen at most $n-1$ times,
so all deletions that result to newly unreachable vertices are handled in $O(mn)$ total time.
Now we consider the cost of executing Algorithm \textsf{DeleteEdge}
when $x$  is reachable and $(x,y)$ is not a bridge.
Lines 8--12 can be implemented in $O(n)$ time.
Also, we can compute $z$ and $c$ in lines 14 and 16, remove the affected vertices from $\mathit{DerivedOut}(c)$ in line 24, and compute the arrays $\mathit{preorder}'$ and $\mathit{size}'$ in line 28 in $O(n)$ time.
So, not accounting for lines 25--27 and for the total running time of Procedures \ref{Procedure:UpdateInSiblings} and \ref{Procedure:LocateNewParent},
Algorithm \textsf{DeleteEdge} executes the sequence of deletions in $O(mn)$ time.
It remains to bound the running time of lines 25--27, and 
of Procedures \ref{Procedure:UpdateInSiblings} and \ref{Procedure:LocateNewParent} by $O(mn)$.

We first bound the time of Procedure~\ref{Procedure:UpdateInSiblings}$(w)$.
Line 1 takes $O(|D(w)| + \mathit{degree}_0(D(w))) = O(\mathit{degree}_0(D(w)))$ time. Set $\mathit{Siblings}(w)$ contains at most $\mathit{degree}_0(D(w))$ vertices.
To perform the test in line 3 in $O(1)$ time, we mark each vertex in $\mathit{DerivedOut}(c)$, since this list only grows during the execution of \textsf{DeleteEdge} for $(x,y)$.
(We unmark all vertices after the deletion.)
We can also test if $d(u) \in \mathit{In}(u)$ in constant time, by maintaining throughout the whole sequence of deletions, a Boolean array $\mathit{DomEdge}$ such that
$\mathit{DomEdge}[u] = \mathit{true}$ if and only $(d(u),u) \in E$.
Finally, inserting a vertex into $Q$ and into a $\mathit{DerivedOut}$ list takes $O(1)$ time.
Thus, Procedure~\ref{Procedure:UpdateInSiblings}$(w)$ is executed in $O(\mathit{degree}_0(D(w)))$ time.
Since the depth of each vertex in $D(w)$ increases by at least one, this running time is at most $\beta ( \Phi'(D(w)) - \Phi(D(w)) )$, for an appropriate constant $\beta$.

Now we bound the time of Procedure~\ref{Procedure:LocateNewParent}$(w)$.
Each execution of the \textbf{foreach} loop takes $O(\mathit{degree}_0(w))$ time.
If the loop is executed $k$ times, then the depth of $w$ increases by $k-1$.
Hence the running time is bounded by $\beta ( \phi'(w) - \phi(w) )$.

We conclude that the total running time of Procedures \ref{Procedure:UpdateInSiblings} and \ref{Procedure:LocateNewParent} is bounded by the
total increase in the potential, which is $O(mn)$.

Finally, we turn to lines 25--27. Let $A$ be the set of affected vertices. 
Set $S$ consists of at most $\mu = \mathit{degree}_0 (A)$ edges. By Lemma \ref{lemma:derived}, we can compute 
$\overline{S}$ in $O(n+\mu)$ time. Again we note that for each affected vertex $v$,  $\mathit{degree}_0(v) \le \phi'(v)-\phi(v)$, so $\mu \le \Phi'(A) - \Phi(A)$,
which implies that the total running time for lines 25--27 is $O(mn)$.
\end{proof}

\section{Conditional Lower Bound}

In the following we give a conditional lower bound for the partially dynamic dominator tree problem.
We show that there is no incremental nor decremental algorithm for maintaining the dominator tree that has total update time $ O ((m n)^{1-\epsilon}) $ (for some constant $ \epsilon > 0 $) unless the \textsf{OMv} Conjecture~\cite{HenzingerKNS15} fails.
This also holds for algorithms that do not explicitly maintain the tree, but are able to answer parent-queries.
Formally, this section contains the proof of the following statement.

\begin{theorem}
For any constant $ \delta \in (0, 1/2] $ and any $ n $ and $ m = \Theta (n^{1 / (1-\delta)}) $, there is no algorithm for maintaining a dominator tree under edge deletions/insertions allowing queries of the form ``is $ x $ the parent of $ y $ in the dominator tree'' that uses polynomial preprocessing time, total update time $ u (m, n) = (m n)^{1-\epsilon} $ and query time $ q (m) = m^{\delta-\epsilon} $ for some constant $ \epsilon > 0 $, unless the \textsf{OMv} conjecture fails.
\end{theorem}

Under this conditional lower bound, the running time of our algorithm is optimal up to sub-polynomial factors.
We give the reduction for the decremental version of the problem.
Hardness of the incremental version follows analogously.

\subsection{Hardness Assumption}

In the online Boolean matrix-vector problem we are first given a Boolean $ n \times n $ matrix~$ M $ to preprocess.
After the preprocessing, we are given a sequence of $n$-dimensional Boolean vectors $ v^{(1)}, \dots, v^{(n)} $ one by one.
For each $ 1 \leq t \leq n $, we have to return the result of the matrix-vector multiplication $ M v^{(t)} $ before we are allowed to see the next vector $ v^{(t+1)} $.
The \textsf{OMv} Conjecture states that there is no algorithm that computes each matrix-vector product correctly (with high probability) and in total spends time $ O (n^{3-\epsilon}) $ for some constant $ \epsilon > 0 $.

We will not use the \textsf{OMv} problem directly as the starting point of our reduction.
Instead we consider the following \textsf{$\gamma$-OuMv} problem (for a fixed $ \gamma > 0 $) and parameters $ n_1 $, $ n_2 $, and $ n_3 $ such that $ n_1 = \lfloor n_2^\gamma \rfloor $:
We are first given a Boolean $ n_1 \times n_2 $ matrix~$ M $ to preprocess. 
After the preprocessing, we are given a sequence of pairs of $n_1$-dimensional Boolean vectors $ (u^{(1)}, v^{(1)}), \dots, (u^{(n_3)}, v^{(n_3)}) $ one by one.
For each $ 1 \leq t \leq n_3 $, we have to return the result of the Boolean vector-matrix-vector multiplication $ (u^{(t)})^\intercal M v^{(t)} $ before we are allowed to see the next pair of vectors $ (u^{(t+1)}, v^{(t+1)}) $.
It has been shown~\cite{HenzingerKNS15} that under the \textsf{OMv} Conjecture as stated above, there is no algorithm for this problem that has polynomial preprocessing time and for processing all vectors spends total time $ O (n_1^{1-\epsilon_1} n_2^{1-\epsilon_2} n_3^{1-\epsilon_3}) $ such that all $ \epsilon_i $ are $ \geq 0 $ and at least one $ \epsilon_i $ is a constant $ > 0 $.

\subsection{Reduction}

We now give the reduction from the \textsf{$\gamma$-OuMv} problem with $ \gamma = \delta / (1-\delta)$ to the decremental dominator tree problem.
In the following we denote by $ v_i $ the $i$-th entry of a vector $ v $ and by $ M_{i,j} $ the entry at row~$ i $ and column~$ j $ of a matrix~$ M $.

Consider an instance of the \textsf{$\gamma$-OuMv} problem with parameters $ n_1 = m^{1-\delta} $, $ n_2 = m^{\delta} $, and $ n_3 = m^{1-\delta} $.
We preprocess the matrix~$ M $ by constructing a graph $ G^{(0)} $ with the set of vertices
\begin{equation*}
V = \{ s, x_1, \dots, x_{n_3}, x_{n_3+1}, y_1, \dots, y_{n_1}, z_1, \dots, z_{n_2} \}
\end{equation*}
and the following edges:
\begin{itemize}
\item an edge $ (s, x_1) $, and, for every $ 1 \leq t \leq n_3 $, an edge $ (x_t, x_{t+1}) $ (i.e.\ a path from $ s $ to $ x_{n_3+1} $)
\item for every $ 1 \leq j \leq n_2 $, an edge $ (x_{n_3+1}, z_j) $
\item for every $ 1 \leq t \leq n_3 $ and every $ 1 \leq i \leq n_1 $, an edge $ (x_t, y_i) $ (i.e., the complete bipartite graph between $ \{ x_1, \ldots, x_{n_3} \} $ and $ \{ y_1, \ldots, y_{n_1} \} $)
\item for every $ 1 \leq i \leq n_1 $ and every $ 1 \leq j \leq n_2 $, an edge $ (y_i, z_j) $ if and only if $ M_{i,j} = 1 $ (i.e.\ a bipartite graph between $ \{ y_1, \ldots, y_{n_1} \} $ and $ \{ z_1, \ldots, z_{n_2} \} $ encoding the matrix~$ M $ in the natural way).
\end{itemize}

Whenever the algorithm is given the next pair of vectors $ (u^{(t)}, v^{(t)}) $, we first create a graph $ G^{(t)} $ by performing the following edge deletions in $ G^{(t-1)} $:
If $ t \geq 2 $, we first delete all outgoing edges of $ x_{t-1} $, except the one to $ x_t $.
Then (for any value of $ t $), for every $ i $ such that $ u^{(t)}_i = 0 $ we delete the edge from $ x_t $ to $ y_i $.
Thus, for every $ 1 \leq i \leq n_1 $, there will be an edge from $ x_t $ to $ y_i $ in $ G^{(t)} $ if and only if $ u^{(t)}_i = 1 $.
Having created $ G^{(t)} $, we now, for every $ j $ such that $ v^{(t)}_j = 1 $, check whether $ x_t $ is the parent of $ z_j $ in the dominator tree.
If this is the case for at least one $ j $ we return that $ (u^{(t)})^\intercal M v^{(t)} $ is $ 1 $, otherwise we return $ 0 $.

\paragraph{Correctness.}

The correctness of our reduction follows from the following lemma.

\begin{lemma}\label{lem:technical lemma lower bound}
For every $ 1 \leq t \leq n $, the $j$-th entry of $ (u^{(t)})^\intercal M $ is $ 1 $ if and only if $ x_{t} $ is the immediate dominator of $ z_j $ in $ G^{(t)} $
\end{lemma}

\begin{proof}
If the $j$-th entry of $ (u^{(t)})^\intercal M $ is $ 1 $, then there is an $ i $ such that $ u^{(t)}_i = 1 $ and $ M_{i,j} = 1 $.
Thus, $ G^{(t)} $ contains the edges $ (x_t, y_i) $ and $ (y_i, z_j) $ and consequently a path from $ s $ to $ z_j $, namely $ \langle s, x_1, \ldots, x_t, y_i, z_j \rangle $.
Vertices that are not on this path cannot be dominators of $ z_j $.
Furthermore, $ y_i $ can also not be a dominator of $ z_j $ because there is a path from $ s $ to $ z_j $ not containing $ y_i $, namely $ \langle s, x_1, \ldots, x_{n_3}, x_{n_3 + 1}, z_j \rangle $.
For every $ 1 \leq t' \leq t-1 $, the vertex $ x_{t'} $ has only one outgoing edge, which goes to $ x_{t'+1} $, as all other outgoing edges are not present in $ G^{(t)} $ anymore.
Thus, all paths from $ s $ to~$ z_j $ necessarily contain the vertices $ s, x_1, \ldots, x_t $, in this order.
Therefore $ x_t $ is the immediate dominator of $ z_j $ in $ G^{(t)} $.

If the $j$-th entry of $ (u^{(t)})^\intercal M $ is $ 0 $, then there is no $ i $ such that $ u^{(t)}_i = 1 $ and $ M_{i,j} = 1 $.
This implies that there is no path (of length $ 2 $) from $ x_t $ to $ z_j $ avoiding $ x_{t+1} $ (via some vertex $ y_i $).
Thus, every path from $ s $ to $ z_j $ contains $ x_{t+1} $, and in particular $ x_{t+1} $ appears after $ x_t $ on such a path.
Thus, $ x_t $ cannot be the immediate dominator of $ z_j $ in $ G^{(t)} $.
\end{proof}

Note that $ (u^{(t)})^\intercal M v^{(t)} $ is $ 1 $ if and only if there is a $ j $ such that both the $j$-th entry of $ u^{(t)} M $ as well as the $j$-th entry of $ v^{(t)} $ are $ 1 $.
Furthermore, $ x_t $ is the parent of $ z_j $ in the dominator tree if and only if $ x_t $ is an immediate dominator of $ z_j $ in the current graph.
Therefore the lemma establishes the correctness of the reduction.

\paragraph{Complexity.}

The initial graph $ G^{(0)} $ has $ n := \Theta (n_1 + n_2 + n_3) = \Theta (m^\delta + m^{1-\delta}) = \Theta (m^{1-\delta}) $ vertices and $ \Theta (n_1 n_2 + n_2 n_3) = \Theta (m) $ edges.
The total number of parent-queries is $ O (n_1 n_3) = m^{2 (1-\delta)} $.
Suppose the total update time of the decremental dominator tree algorithm is $ O (u(m, n)) = (m n)^{1-\epsilon} $ and its query time is $ O (q(m)) = m^{\delta - \epsilon} $.
Using the reduction above, we can thus solve the \textsf{$\gamma$-OuMv} problem for the parameters $ n_1, n_2, n_3 $ with polynomial preprocessing time and total update time
\begin{equation*}
O (u (m, n) + m^{2 (1-\delta)} q (m)) = O (u (m, m^{1-\delta}) + m^{2 (1-\delta)} q (m)) = O (m^{2-\delta - \epsilon}) \, .
\end{equation*}
Since $ n_1 n_2 n_3 = m^{2-\delta} $, this means we would get an algorithm for the \textsf{$\gamma$-OuMv} problem with polynomial preprocessing time and total update time $ O (n_1^{1-\epsilon_1} n_2^{1-\epsilon_2} n_3^{1-\epsilon_3}) $ where at least one $ \epsilon_i $ is a constant $ > 0 $.
This contradicts the \textsf{OMv} Conjecture.


%
%

\bibliographystyle{plain}
\bibliography{ltg,references}

\end{document}